\newcommand{\ShoLong}[2]{#1} 

\documentclass{article}
\usepackage{fullpage}
\usepackage{amssymb}
\usepackage{amsmath}
\usepackage{amsthm}
\usepackage{graphicx}
\usepackage{subfigure}
\usepackage{wrapfig}
\usepackage{xspace}
\usepackage{todonotes}
\usepackage[pagewise]{lineno}
\usepackage{soul}
\usepackage{paralist}
\usepackage{float}
\usepackage{hyperref}

\providecommand{\keywords}[1]{{\small \textbf{\textit{Keywords:}} #1}}

\newcommand{\oneblip}{$1$-gap-planar\xspace}
\newcommand{\twoblip}{$2$-gap-planar\xspace}
\newcommand{\kblip}{$k$-gap-planar\xspace}
\newcommand{\kblipa}{$k$-gap assignment\xspace}

\newtheorem{theorem}{Theorem}
\newtheorem{lemma}[theorem]{Lemma}
\newtheorem{property}[theorem]{Property}
\newtheorem{corollary}[theorem]{Corollary}

\usepackage{etoolbox}

\newif\ifComments
\Commentstrue

\usepackage{xcolor}
\ifComments
  \newcommand{\says}[3][red]{\textcolor{#1}{\textsc{#2 says:} \textsf{#3}}}
\else
  \newcommand{\says}[2][red]{\relax}
\fi

\title{Gap-planar Graphs\thanks{A preliminary version of this paper appeared in the proceedings of the {\em 25th International Symposium on Graph Drawing}~\cite{bbceeghkmrt-gpg-17}.}}

\author{
Sang Won Bae\thanks{Kyonggi University, Suwon, South Korea. \texttt{swbae@kgu.ac.kr}. Supported by the Basic Science Research Program through the National Research Foundation of Korea (NRF) funded by the Ministry of Education (2015R1D1A1A01057220).}
\and
Jean-Francois Baffier\thanks{Tokyo Institute of Technology, Tokyo, Japan. \texttt{jf\_baffier@nii.ac.jp}. J.-F.B. was partially supported by JST ERATO Grant Number JPMJER1305}
\and
Jinhee Chun\thanks{Tohoku University, Sendai, Japan.
\texttt{\{jinhee,mati\}@dais.is.tohoku.ac.jp}. M.K. was partially supported by MEXT KAKENHI No.~15H02665, 17K12635 and JST ERATO Grant Number JPMJER1305.}
\and
Peter Eades\thanks{University of Sydney, Sydney, Australia.
{\tt peter.eades@sydney.edu.au, shhong@it.usyd.edu.au}. Partially supported by ARC DP160104148.}
\and
Kord Eickmeyer\thanks{TU Darmstadt, Darmstadt, Germany. \texttt{eickmeyer@mathematik.tu-darmstadt.de}.}
\and
Luca Grilli\thanks{University of Perugia, Perugia, Italy. \texttt{\{luca.grilli,fabrizio.montecchiani\}@unipg.it}
}
\and
Seok-Hee Hong\footnotemark[5]
\and
Matias Korman\footnotemark[4]
\and
Fabrizio Montecchiani\footnotemark[7]
\and
Ignaz Rutter
\thanks{University of Passau, Passau, Germany. {\tt rutter@fim.uni-passau.de}}
\and
Csaba D. T{\'o}th\thanks{California State University Northridge, Los Angeles, CA,
and Tufts University, Medford, MA, USA. \texttt{csaba.toth@csun.edu}. Partially supported by the NSF awards CCF-1422311 and CCF-1423615.}
}

\begin{document}

\pagestyle{plain}

\date{}
\maketitle

\begin{abstract}
 We introduce the family of \emph{\kblip graphs} for $k \geq 0$, i.e., graphs that have a drawing in which each crossing  is assigned to one of the two involved edges and each edge is assigned at most $k$ of its crossings. This definition is motivated by applications in edge casing, as a \kblip graph can be drawn crossing-free after introducing at most $k$ local gaps per edge. We present results on the maximum density of \kblip graphs, their relationship to other classes of beyond-planar graphs, characterization of  \kblip complete graphs, and the computational complexity of recognizing \kblip graphs.
\end{abstract}


\section{Introduction}
\begin{figure}[t]
 \centering
 \subfigure{\includegraphics[width=0.22\columnwidth,page=1]{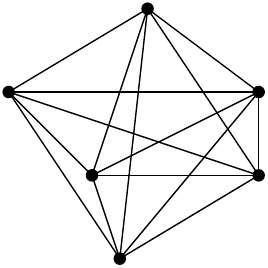}\label{fig:drawing}}
 \hfil
 \subfigure{\includegraphics[width=0.22\columnwidth,page=2]{figs/example}\label{fig:cased}}
 \caption{\label{fig:example} A drawing of a graph $G$ (left)
 and its cased version where each edge is interrupted at most twice, i.e., a \twoblip drawing of $G$ (right).}
\end{figure}

Minimizing the overall number of edge crossings in a drawing has been the main objective of a large body of literature concerning the design of algorithms to automatically draw a graph. In fact, several graph drawing algorithms assume the input graph to be planar or planarized (that is, crossings are replaced with dummy vertices which are removed in a post-processing step). More recently, cognitive experiments suggested that the absence of specific kinds of edge crossing configurations has a positive impact on the human understanding of a graph drawing~\cite{DBLP:journals/vlc/HuangEH14}. These practical findings motivated a line of research, commonly called \emph{beyond planarity}, whose focus is on non-planar graphs that can be drawn by locally avoiding specific edge crossing configurations or by guaranteeing specific properties for the edge crossings (see, e.g.,~\cite{JGAA-459,hong_et_al,shonan,DBLP:conf/ictcs/Liotta14}).

Among the most investigated families of beyond-planar graphs are: \emph{$k$-planar graphs} (see, e.g.,~\cite{DBLP:conf/gd/Bekos0R16,DBLP:journals/csr/KobourovLM17,DBLP:journals/combinatorica/PachT97}),
which can be drawn with at most $k$ crossings per edge; \emph{$k$-quasiplanar graphs} (see, e.g.,~\cite{DBLP:journals/jct/AckermanT07,DBLP:journals/combinatorica/AgarwalAPPS97,DBLP:journals/siamdm/FoxPS13}), which can be drawn with no $k$ pairwise crossing edges; \emph{fan-planar graphs} (see, e.g.,~\cite{Bekos2016,DBLP:journals/tcs/BinucciGDMPST15,DBLP:journals/corr/KaufmannU14}),
which can be drawn such that each edge is crossed by a (possibly empty) set of edges that have a common endpoint 
on one side;
\emph{RAC graphs} (refer, e.g., to~\cite{dl-cargd-12}), which admit a straight-line drawing with right-angle crossings.

In this paper we introduce a family that generalizes $k$-planar graphs by introducing a nonsymmetric constraint on the intersection pattern of the edges. Intuitively speaking, we charge each crossing to only one of the two edges involved in the crossing and do not allow an edge to be charged many times. This constraint is motivated by \emph{edge casing}, a method commonly used to alleviate the visual clutter generated by crossing lines in a diagram~\cite{Appel:1979:HLE:965103.807437,DBLP:journals/comgeo/EppsteinKMS09}. In a \emph{cased drawing} of a graph, each crossing is resolved by locally interrupting one of the two crossing edges; see Figure~\ref{fig:example} for an illustration. This edge casing makes only \emph{one} of the edges involved in the crossing hard to follow whereas the other one is unaffected. Regardless of the number of crossings, the drawing will remain clear as long as no edge is cased many times; thus, an edge could participate in arbitrarily many crossings as long as the other edges are cased. Eppstein et al.~\cite{DBLP:journals/comgeo/EppsteinKMS09} studied several optimization problems related to edge casing, assuming the input is a graph together with a fixed drawing. In particular, the problem of minimizing the maximum number of gaps per edge in a drawing can be solved in polynomial time (see also Section~\ref{sec:preliminaries}). We also note that a similar drawing paradigm is used by \emph{partial edge drawings (PEDs)}, in which the central part of each edge is erased, while the two remaining stubs are required to be crossing-free (see, e.g.,~\cite{JGAA-438,DBLP:conf/fun/BruckdorferK12}).

We formalize this idea with the family of \emph{\kblip graphs}, a family of graphs that can be drawn in the plane such that each crossing is assigned to one of the two involved edges and each edge is assigned at most $k$ crossings (for some constant $k$).
We present a rich set of results for \kblip graphs related to classic research questions, such as bounds on the maximum density, drawability of complete graphs, complexity of the recognition problem, and relationships with other families of beyond-planar graphs. Our results can be summarized as follows:
\begin{itemize}\itemsep -2pt
 \item Every \kblip graph with $n$ vertices has $O(\sqrt{k} \cdot n)$ edges (Section~\ref{sec:density}). If $k=1$, we prove an upper bound of $5n-10$ for the number of edges in a \oneblip multigraph with $n$ vertices (without homotopic parallel edges), and construct \oneblip (simple) graphs that attain this bound for all $n\geq 20$. Note that the same density bound is known to be tight for $2$-planar graphs~\cite{DBLP:journals/combinatorica/PachT97}.

 \item We study relationships between the class of \kblip graphs and other classes of beyond-planar graphs. For all $k \ge 1$, the class of $2k$-planar graphs is properly contained in the class of \kblip graphs, which in turn is properly contained in the $(2k+2)$-quasiplanar graphs (Section~\ref{sec:relationship}).
     We note that $k$-planar graphs are known to be  $(k+1)$-quasiplanar~\cite{DBLP:journals/corr/AngeliniBBLBDLM17,DBLP:journals/corr/HoffmannT17}. Furthermore, we investigate the relationship between \kblip graphs and $d$-degenerate crossing graphs, a class of graphs recently introduced by Eppstein and Gupta~\cite{eg-cpnrn-17}.

 \item The complete graph $K_n$ is \oneblip if and only if $n \le 8$ (Section~\ref{sec:complete}).

 \item Deciding whether a graph is \oneblip is \textsc{NP}-complete, even when the drawing of a given graph is
     restricted to a fixed rotation system that is part of the input (Section~\ref{sec:recognition}). Note that analogous recognition problems for other families of beyond-planar graphs are also \textsc{NP}-hard (see, e.g.,~\cite{DBLP:journals/jgaa/AuerBGR15,Bekos2016,DBLP:journals/tcs/BinucciGDMPST15,DBLP:journals/tcs/BrandenburgDEKL16,DBLP:journals/algorithmica/GrigorievB07,DBLP:journals/jgt/KorzhikM13}), while polynomial algorithms are known in some restricted settings (see, e.g., \cite{DBLP:journals/algorithmica/AuerBBGHNR16,Bekos2016,DBLP:journals/tcs/BrandenburgDEKL16,DBLP:journals/tcs/DehkordiEHN16,DBLP:journals/tcs/EadesHKLSS13,DBLP:conf/wg/HongN15,DBLP:journals/algorithmica/HongEKLSS15}).

\end{itemize}

Preliminaries and basic results are in Section~\ref{sec:preliminaries}. Conclusions and open problems are discussed in Section~\ref{sec:conclusions}.

\section{Preliminaries and basic results}
\label{sec:preliminaries}
A \emph{drawing} $\Gamma$ of a graph $G=(V,E)$ is a mapping of the
vertices of $V$ to distinct points, and of the edges of
$E$ to a continuous arcs connecting their corresponding endpoints
such that no edge (arc) passes through any vertex,
if two edges have a common interior point in $\Gamma$, then
they cross transversely at that point,
and no three edges cross at the same point.
For a subset $E' \subseteq E$, the restriction of $\Gamma$
to the curves representing the edges of $E'$ is denoted by $\Gamma[E']$.
A drawing $\Gamma$ is \emph{planar} if no two edges cross.
A graph is \emph{planar} if it admits a planar drawing.
A \emph{planar embedding} of a planar graph $G$ is an equivalence
class of topologically equivalent (i.e., isotopic) planar drawings of $G$.
A \emph{plane graph} is a planar graph with a planar embedding.
A planar drawing subdivides the plane into topologically connected regions,
called \emph{faces}. The unbounded region is the \emph{outer face}.

The \emph{crossing number} ${\rm cr}(G)$ of a graph~$G$ is the smallest
number of edge crossings over all drawings of $G$.
The \emph{crossing graph} $C(\Gamma)$ of a drawing $\Gamma$ is the graph having a vertex $v_e$
for each edge $e$ of $G$, and an edge $(v_e,v_f)$ if and only if edges
$e$ and $f$ cross in $\Gamma$. The \emph{planarization} $\Gamma^*$ of
$\Gamma$ is the plane graph formed from $\Gamma$ by inserting a
\emph{dummy vertex} at each crossing, and subdividing
both edges with the dummy vertex.  To avoid ambiguities,
we call \emph{real vertices} the vertices of $\Gamma^*$
that are in $V$ (i.e., that are not dummy).

A class of graphs is informally called ``beyond-planar'' if the graphs in this family admit drawings in which the intersection patterns of the edges are characterized by some forbidden configuration (see, e.g.,~\cite{hong_et_al,shonan,DBLP:conf/ictcs/Liotta14}).
Research on such graph classes is attracting increasing attention in graph theory, graph algorithms, graph drawing, and computational geometry, as these graphs represent a natural generalization of planar graphs, and their study can provide significant insights for the design of effective methods to visualize real-world networks. Indeed, the motivation for this line of research stems  from both the interest raised by the combinatorial and geometric properties of these graphs, and  experiments showing how the absence of particular edge crossing patterns has a positive impact on the readability of a graph drawing~\cite{DBLP:journals/vlc/HuangEH14}.

Among the investigated families of beyond-planar graphs are: \emph{$k$-planar graphs} (see, e.g.,~\cite{DBLP:conf/gd/Bekos0R16,DBLP:journals/csr/KobourovLM17,DBLP:journals/combinatorica/PachT97}), which can be drawn in the plane with at most $k$ crossings per edge; \emph{$k$-quasiplanar graphs} (see, e.g.,~\cite{DBLP:journals/jct/AckermanT07,DBLP:journals/combinatorica/AgarwalAPPS97,DBLP:journals/siamdm/FoxPS13}), which can drawn without $k$ pairwise crossing edges; \emph{fan-planar graphs} (see, e.g.,~\cite{Bekos2016,DBLP:journals/tcs/BinucciGDMPST15,DBLP:journals/corr/KaufmannU14}), which can be drawn such that no edge crosses two independent edges; \emph{fan-crossing-free graphs}~\cite{Cheong2015}, which can be drawn such that no edge crosses any two edges that are adjacent to each other; \emph{planarly-connected graphs}~\cite{DBLP:conf/gd/AckermanKV16}, which can be drawn such that each pair of crossing edges is independent and there is a crossing-free edge that connects their endpoints;  \emph{RAC graphs} (refer, e.g., to~\cite{dl-cargd-12}), which admit a straight-line (or polyline with few bends) drawing where any two crossing edges are perpendicular to each other.

Eppstein et al.~\cite{DBLP:journals/comgeo/EppsteinKMS09} studied several optimization problems related to edge casing, assuming the input is a graph together with a fixed drawing. In particular, the problem of minimizing the maximum number of gaps per edge in a drawing can be solved in polynomial time (see also Section~\ref{sec:preliminaries}). We also note that a similar drawing paradigm is used by \emph{partial edge drawings (PEDs)}, in which the central part of each edge is erased, while the two remaining stubs are required to be crossing-free (see, e.g.,~\cite{JGAA-438,DBLP:conf/fun/BruckdorferK12}).

Let $\Gamma$ be a drawing of a graph $G$. Recall that exactly
two edges of $G$ cross in one point $p$ of $\Gamma$, and we say that
these two edges are \emph{responsible} for $p$. A \emph{\kblipa} of
$\Gamma$ maps each crossing point of $\Gamma$ to one of its two
responsible edges so that each edge is assigned with at most $k$ of its
crossings; see, e.g., 
Fig.~\ref{fig:example}(right).
A \emph{gap} of an edge is a crossing assigned to it. An edge with at least one
gap is \emph{gapped}, else it is \emph{gap-free}.
A drawing is \emph{\kblip} if it admits a \kblipa.
A graph is \emph{\kblip} if it has a \kblip drawing. Note that
a graph is planar if and only if it is $0$-gap-planar, and that
$k$-gap-planarity is a monotone property: every subgraph of a \kblip
graph is \kblip. The summation of the number of gaps over all edges
in a set $E'\subset E$ yields the following.

\begin{property}\label{pr:edges-crossings}
  Let $\Gamma$ be a \kblip drawing of a graph $G=(V,E)$.
  For every $E' \subseteq E$, the subdrawing $\Gamma[E']$ contains at most $k \cdot |E'|$
  crossings.
\end{property}

In fact, the converse of Property~\ref{pr:edges-crossings} also holds,
and we obtain the following stronger result.

\begin{theorem}\label{thm_crossing}
  Let $\Gamma$ be a drawing of a graph $G=(V,E)$.  The drawing
  $\Gamma$ is \kblip if and only if for each edge set
  $E' \subseteq E$ the subdrawing $\Gamma[E']$ contains at most
  $k \cdot |E'|$ crossings.
\end{theorem}

\begin{proof}
  Property~\ref{pr:edges-crossings} is the only-if direction.  It remains to prove the
  forward direction.  Let $A$ denote the set of crossings in $\Gamma$.
  Further let $B = \{ e_1, \dots,e_k : e \in E\}$ denote a set that
  consists of $k$ copies of each edge in $G$.  Let $H$ be the
  bipartite graph whose vertex set is $A \dot\cup B$ and where a
  crossing $p \in A$ is connected to all copies of edges that are
  responsible for the crossing $p$.

  Clearly, \kblip assignments correspond bijectively to matchings $M$
  in $H$ such that each crossing $a \in A$ is incident to an edge in
  $M$.
  By Hall's theorem, the bipartite graph $H$ has a matching from $A$ into $B$
  if and only if for each set $X \subseteq A$, we have $|N(X)| \ge |X|$.

  Let $X \subseteq A$ be some subset of crossings. Let $E(X)$ denote
  the set of edges that are responsible for crossings in $X$.  By
  considering the subdrawing $\Gamma[E(X)]$, we find $|X| \le k|E(X)|$.
  Moreover, by construction of $H$ the neighborhood $N(X)$ of $X$
  contains exactly $k$ vertices for each edge in $E(X)$, i.e.,
  $|N(X)| = k|E(X)|$.  Thus it is $|X| \le |N(X)|$, which is Hall's
  condition.  Thus a \kblip assignment exists, showing that $\Gamma$
  is \kblip.
\end{proof}

\textbf{Note:} 
David Wood (personal communication) has suggested an alternative proof of the above statement:
Hakimi~\cite{hakimi} proved that a graph has an orientation with maximum outdegree at most $k$ if and only if every
subgraph has average degree at most $2k$. Theorem~\ref{thm_crossing} immediately follows by applying this result to 
the intersection graph of the edges in a drawing of a graph. 

A \kblipa of a drawing $\Gamma$ corresponds to orienting the edges of
the crossing graph $C(\Gamma)$ such that each vertex has indegree at
most $k$ (intuitively, orienting a crossing towards an edge
means we assign the crossing to that edge). Since finding an orientation
of a graph with the smallest maximum indegree corresponds to finding its
pseudoarboricity~\cite{Frank1976,Picard1982},
Property~\ref{pr:pseudoarboricity} below follows. A
\emph{pseudoforest} is a graph in which every connected component has
at most one cycle, and the \emph{pseudoarboricity} of a graph is
the smallest number of pseudoforests needed to cover all its edges.

\begin{property}\label{pr:pseudoarboricity}
  A graph is \kblip if and only if it admits a drawing whose crossing
  graph has pseudoarboricity at most $k$.
\end{property}
Given a drawing $\Gamma$ of a graph $G=(V,E)$, we can find the minimum $k\geq 0$
such that $\Gamma$ is \kblip in $O(|E|^4)$ time, due to the fact that one can find
an orientation of $C(\Gamma)$ with the smallest maximum indegree
in time quadratic in the number of edges of $C(\Gamma)$~\cite{Venkateswaran2004374}.

\textbf{Note:} 
In an earlier versions of this paper~\cite{bbceeghkmrt-gpg-17,blipJournal} we gave an upper bound on the treewidth of \kblip graphs. Our bounds were based on a result by Dujmovi{\'c}, Eppstein, and Wood~\cite{SIAM-OTDG} that bounded the treewidth of a graph as a function of the number of vertices and the crossing number. Unfortunately, their result turned out to be incorrect~\cite{xxx}. Thus, it still remains open to show that \kblip graphs have $O(f(k)\sqrt{n})$ treewidth for some function $f$.





\section{Density of \kblip graphs}
\label{sec:density}

We begin with an upper bound on the number of edges of \kblip graphs.

\begin{theorem}\label{thm:kdensity}
 A \kblip graph on $n \ge 3$ vertices has $O(\sqrt{k}\cdot n)$ edges.
\end{theorem}
\begin{proof}
 The crossing number of a graph $G$ with $n$ vertices and $m$ edges is bounded by ${\rm cr}(G)\geq \frac{1024}{31827}\cdot m^3/n^2$ when $m\geq \frac{103}{6}n$~\cite{DBLP:journals/dcg/PachRTT06}. Combined with the bound ${\rm cr}(G)\leq k \cdot m$ (Property~\ref{pr:edges-crossings}), we obtain
 $$\frac{1024}{31827}\cdot  \frac{m^3}{n^2}\leq {\rm cr}(G)\leq km,$$
 which implies $m\leq \max(5.58 \sqrt{k},17.17)\cdot n$, as required.
\end{proof}

Better upper bounds are possible for small values of $k$, in particular for $k=1$. Pach et al.~\cite{DBLP:journals/dcg/PachRTT06} proved that a graph $G$ with $n\geq 3$ vertices satisfies ${\rm cr}(G)\geq \frac{7}{3}m-\frac{25}{3}(n-2)$. Combined with the bound ${\rm cr}(G)\leq k \cdot m$, we have $$m\leq \frac{25(n-2)}{7-3k}.$$
For $k=1$ (i.e., for \oneblip graphs), this gives $m\leq 6.25n-12.5$. We now show how to improve this bound to $m\leq 5n-10$ (see Theorem~\ref{thm:density} below). The idea is to follow a strategy developed by Pach and T\'oth~\cite{DBLP:journals/combinatorica/PachT97} and Bekos et al.~\cite{DBLP:conf/gd/Bekos0R16} on the density of 2- and 3-planar graphs, with several important differences.

In order to accommodate the elementary operations in the proof of Theorem~\ref{thm:density}, we work on a broader class of graphs. A drawing $\Gamma$ of a multigraph $G=(V,E)$ is \kblip if it admits a \kblipa and no two parallel edges are homotopic. A multigraph is \emph{\kblip} if it has a \kblip drawing.
Two parallel edges $e_1=e_2=(u,v)$ are \emph{homotopic} in a drawing $\Gamma$, if the drawings are continuous arcs $\gamma_1:[0,1]\rightarrow \mathbb{R}^2$ and $\gamma_1:[0,1] \mathbb{R}^2$,
and there is a continuous function $h:[0,1]^2\rightarrow \mathbb{R}^2$ such that $\gamma_1(t)=h(0,t)$, $\gamma_2(t)=h(1,t)$, $\Gamma(u)=h(t,0)$, and $\Gamma(v)=h(t,1)$ for all $t\in [0,1]$, and $h$ does not map any point of the open square $(0,1)^2$ to a vertex in $\Gamma$. Intuitively, $\gamma_1$ can be continuously deformed into $\gamma_2$ with fixed endpoints and without passing through any vertex in $\Gamma$.
In particular, in a \oneblip drawing, two homotopic parallel edges either cross no other edge (they might cross each other), or they both cross the same edge.

Let $n\in \mathbb{N}$, $n\geq 3$. Let $G=(V,E)$ be a \oneblip multigraph with $n$ vertices that has the maximum number of edges possible over all $n$-vertex \oneblip multigraphs; and let $\Gamma$ be a \oneblip drawing of $G$.
Let $H=(V,E')$ be a sub-multigraph of $G$, where $E'\subseteq E$ is a maximum multiset of edges that are pairwise noncrossing in $\Gamma[E']$, and if there are several such sub-mustligraphs, then $H$ has the fewest connected components.

\ShoLong{
 We first show that $H=(V,E')$ is a \emph{triangulation}, that is, a plane multigraph in which every face is bounded by a walk with three vertices and three edges.
}{
 Our proof is based on the next technical lemma.
}

\begin{lemma}\label{lem:H}
 \ShoLong{
  The multigraph $H$ is a triangulation.
 }{
  The multigraph $H$ is a triangulation, that is, a plane multi-graph in which every face is bounded by a walk with three vertices and three edges.
 }
\end{lemma}

\ShoLong{
The proof of Lemma~\ref{lem:H} is deferred to Section~\ref{ssec:H}.
}{
}
We can now show that $|E|\leq 5n-10$. We state a stronger result (for multigraphs)
that immediately implies the same density bound for simple graphs (Corollary~\ref{cor:density}).

\begin{theorem}\label{thm:density}
A multigraph on $n \ge 3$ vertices that has a \oneblip drawing in which no two parallel edges are homotopic
has at most $5n-10$ edges.
\end{theorem}
\begin{corollary}\label{cor:density}
 A \oneblip (simple) graph on $n \ge 3$ vertices has at most $5n-10$ edges.
\end{corollary}
\begin{proof}[Proof of Theorem~\ref{thm:density}.]
 By Lemma~\ref{lem:H}, $H=(V,E')$ is a triangulation. By Euler's polyhedron theorem, it has $3n-6$ edges and $2n-4$ triangular faces. Consider the edges in $E''=E\setminus E'$. It remains to show that $|E''|\leq 2n-4$.

 The embedding of edge $e\in E''$ is a Jordan arc that visits two or more triangle faces of $H$. We call the first and last triangles along $e$ the \emph{end triangles} of $e$. For an end triangle $\Delta$, the connected component of $e\cap \Delta$ incident to a vertex of $\Delta$ is called an \emph{end portion}.
 We use the following charging scheme.

 \begin{figure}[htbp]
 \centering
 \includegraphics[width=0.45\columnwidth]{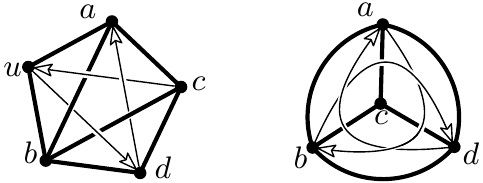}
 \caption{Example for the charging scheme in the proof of Theorem~\ref{thm:density}.
  Bold edges are in a crossing-free triangulation $H=(V,E')$.
  Every edge in $E\setminus E'$ is charged to a triangle of $H$ as indicated by arrows.
  Left: a simple graph where edge $(c,u)$ is charged to $\Delta abu$, edge $(d,u)$ to $\Delta acd$, and edge $(a,d)$ to $\Delta abc$.
  Right: A \oneblip multigraph with nonhomotopic parallel edges.}\label{fig:charge}
\end{figure}

 Each edge $e\in E''$ charges one unit to a triangle face of $H$ as follows. (Refer to Fig.~\ref{fig:charge}.)
 If $e$ has an end portion that has a gap neither in the interior nor on the boundary of the corresponding end triangle $\Delta$, then $e$ charges one unit to $\Delta$. (If neither end portions of $e$ has a gap in the interior or on the boundary of its end triangle, then $e$ charges one arbitrary end triangle.) Otherwise the two end portions of $e$ lie in two adjacent triangles, say, $\Delta_1$ and $\Delta_2$, and $e$ uses its gap to cross an edge $e'$ on the boundary between $\Delta_1$ and $\Delta_2$; in this case $e$ charges one unit to $\Delta_1$ or $\Delta_2$ as follows: If $e'$ has a gap and the edge $e''\in E''$ passing through this gap charges $\Delta_1$ (because $e''$ has an end portion $e''\cap \Delta_1$ that has a gap neither in the interior nor on the boundary of $\Delta_1$), then $e$ charges $\Delta_2$, otherwise it charges $\Delta_1$.

 We claim that each face of $H$ receives at most one unit of charge. Let $\Delta=\Delta{abc}$ be a face in $H$.
 Suppose to the contrary that $\Delta$ receives positive charge from two edges, say $e_1,e_2\in E''$.
 Then both edges have an end portion in $\Delta$ that do not have gaps in the interior of $\Delta$. Consequently, the end portions of $e_1$ and $e_2$ in $\Delta$ cannot cross, and so they are incident to the same vertex of $\Delta$. Therefore, the both end portions $e_1\cap \Delta$ and $e_2\cap \Delta$ are incident to the same vertex of $\Delta$, say $a$, and cross the edge of $\Delta$ opposite to $a$, namely $(b,c)$. Let $\Delta'=\Delta'bcd$ be the face of the plane graph $H$ on the opposite side of $(b,c)$.

  Assume first that the end portion $e_1\cap \Delta$ has a gap neither in the interior nor on the boundary of $\Delta$. Then $e_1$ passes through the gap of $(b,c)$. Since $(b,c)$ has at most one gap in a \oneblip drawing, $e_2$ uses its own gap to cross $(b,c)$. By our charging scheme, this implies that $e_2=(a,d)$, and it must charges one unit to $\Delta'$ (rather than $\Delta$). Next assume that $(b,c)$ does not have any gap.
  Then $e_1$ and $e_2$ each use their own gaps to cross $(b,c)$. Both $e_1$ and $e_2$ are homotopic to an edge $(a,d)$ lying in $\Delta_1\cup \Delta_2$ by our charging scheme. All cases lead to a contradiction, hence  $\Delta$ receives at most $1$ unit of charge, as claimed. Consequently, $|E''|$ is bounded above by the number of faces of $H$, which is $2n-4$, as required.
\end{proof}

\newcommand{\proofLemmaOneSection}{
\subsection{Proof of Lemma~\ref{lem:H}}\label{ssec:H}

We start with a few basic observations.
Let $G$ be an edge-maximal multigraph on $n$ vertices that
has a \oneblip drawing without homtopic parallel edges.

\begin{lemma}\label{lem:g-con}
 Graph $G=(V,E)$ is connected.
\end{lemma}
\begin{proof}
 Suppose, to the contrary, that $G$ is disconnected. Let $G_1=(V_1,E_1)$ be one component, and let $G_2=(V_2,E_2)$ be the disjoint union of all other components (i.e., $V_2=V\setminus V_1$ and $E_2=E\setminus E_1$). For $i=1,2$, let $\Gamma_i=\Gamma[E_i]$ (i.e., the drawing of $G_i$ inherited from $G$), and let $\Gamma^*_i$ be the  planarization of $\Gamma_i$.

 Let $f_2$ be a face in $\Gamma^*_2$ incident to some vertex $v_2\in V_2$. Apply a projective transformation to $\Gamma_1$ so that the outer face is incident to some vertex $v_1\in V_1$; followed by an affine transformation that maps $\Gamma_1$ into the interior of face $f_2$. We obtain a \oneblip drawing of $G$ in which we can insert a new crossing-free edge $(v_1,v_2)$, between two distinct components of $G$, contradicting the maximality of $G$.
\end{proof}

Recall that $\Gamma$ is a \oneblip drawing of $G$ with the minimum number of crossings. We show that this implies that $\Gamma$ is a \emph{simple topological} drawing, that is, no edge crosses itself and every pair of edges cross at most once. This follows from standard simplification techniques, but we provide the proof for completeness.
\begin{lemma}\label{lem:simpletop}
$\Gamma$ is a simple topological drawing.
\end{lemma}
\begin{proof}
Suppose the drawing $\gamma_0$ of an edge $e_0=(u,v)$ crosses itself at point $c_0$. Then $\gamma_0$ crosses itself only once, and this crossing is charged to edge $e$, hence all other crossings of $\gamma_0$ are charged to other edges. We can redraw $e$ by eliminating the loop of $\gamma_0$. This yields a new \oneblip drawing of $G$ with at least one fewer crossings, contradicting the minimality of $\Gamma$.

Suppose the drawings $\gamma_1$ and $\gamma_2$ of edges $e_1$ and $e_2$ cross at points $c_1$ and $c_2$.
Then they cross exactly twice and these two crossings are charged to $e_1$ and $e_2$, hence any other crossing of $e_1$ or $e_2$ with some edge $e_3$ is charged to $e_3$. We can redraw $e_1$ and $e_2$ in $\gamma_1\cup \gamma_2$ by exchanging their subarcs between $c_1$ and $c_2$ such that both crossings are eliminated. This yields a new \oneblip drawing of $G$ with fewer crossings, contradicting the minimality of $\Gamma$.
\end{proof}

Since $G$ is connected, every face in the planarization $\Gamma^*$ of $\Gamma$ has a connected boundary. The \emph{boundary walk} of a face $f$ is a closed walk $(a_1,a_2,\ldots, a_m)$ in $\Gamma^*$ such that $f$ lies on the left hand side of each edge $(a_i,a_{i+1})$ along the walk; and every two consecutive edges of the walk, $(a_{i-1},a_i)$ and $(a_i,a_{i+1})$, are also consecutive in the counterclockwise rotation of all edges incident to $a_i$.
Let $F_0$ denote the set of faces in the planarization $\Gamma^*$ that are not incident to any vertex in $V$.

\begin{lemma}\label{lem:cycle}
 If $f\in F_0$, then the boundary walk of $f$ is
 \begin{enumerate}
  \item a simple cycle (i.e., has no repeated vertices) with at least 3 vertices;
  \item disjoint from the boundary walk of any other face in $F_0$.
 \end{enumerate}
\end{lemma}
\begin{proof}
 \noindent 1.
 Let $f\in F_0$, and let $w=(a_1,a_2,\ldots, a_\ell)$ be its boundary walk. By Lemma~\ref{lem:simpletop}, we have $\ell\geq 3$. Let $C_f=\{a_1,\ldots, a_\ell\}$ be the set of vertices in $w$; and let $E_f\subset E$ be the set of edges in $G$ that contain some edge of $w$. It suffices to show that $|C_f|=\ell$, and then $w$ has no repeated vertices, hence it is a simple cycle.

 Suppose, to the contrary, that the vertices in $w$ are not distinct. Since $f\in F_0$, all vertices in $w$ are crossings in the drawing $\Gamma$, consequently they all have degree 4 in the planarization $\Gamma^*$.
 If $a_i=a_j$, $i\neq j$, then $a_i$ and $a_j$ cannot be consecutive vertices in $w$, and two pairs of edges from $(a_{i-1},a_i)$, $(a_i,a_{i+1})$, $(a_{j-1},a_j)$, $(a_j,a_{j+1})$ are part of the same edge in $E$. If $|C_f|=\ell-k$, for some $k\in \mathbb{N}$, then $|E_f|\leq \ell-2k$. This implies $|E_f|<|C_f|$. That is, the edges in $E_f$ are involved in more than $|E_f|$ crossings, contradicting the assumption that $\Gamma$ is a \oneblip drawing.

 \medskip
 \noindent 2. Let $f_1,f_2\in F_0$ be two faces, with boundary walks $w_1=(a_1,\ldots, a_\ell)$ and $w_2=(b_1,\ldots, b_{\ell'})$. Both $w_1$ and $w_2$ are simple cycles by part~1. For $i=1,2$, let $C_i$ be the set of vertices in $w_i$, and $E_i\subseteq E$ the set of edges of $G$ that contain the edges of the walk $w_i$.

 Note that $w_1$ and $w_2$ cannot share two consecutive edges, say $(a_{i-1},a_i)$ and $(a_i,a_{i+1})$, since the middle vertex $a_i$ has degree 4 in $\Gamma^*$. When $w_1$ and $w_2$ have a common edge, say $(a_i,a_{i+1})=(b_{j+1},b_j)$, then three pairs of edges from $(a_{i-1},a_i)$, $(a_i,a_{i+1})$, $(a_{i+1},a_{i+2})$, $(b_{j-1},b_j)$, $(b_j,b_{j+1})$ $(b_{j+1},b_{j+2})$ are part of the same edge in $E$.
 When $w_1$ and $w_2$ have a common vertex $a_i=b_j$ but no common edge incident to $a_i=b_j$,
 then two pairs of edges from $(a_{i-1},a_i)$, $(a_i,a_{i+1})$, $(b_{j-1},b_j)$, $(b_j,b_{j+1})$ are part of the same edge in $E$. This implies $|E_1\cup E_2|<|C_1\cup C_2|$. That is, the edges in $E_1\cup E_2$ are involved in more than $|E_1\cup E_2|$ crossings, contradicting the assumption that $\Gamma$ is \oneblip.
\end{proof}

\begin{figure}[htbp]
 \centering
 \includegraphics[width=0.4\columnwidth]{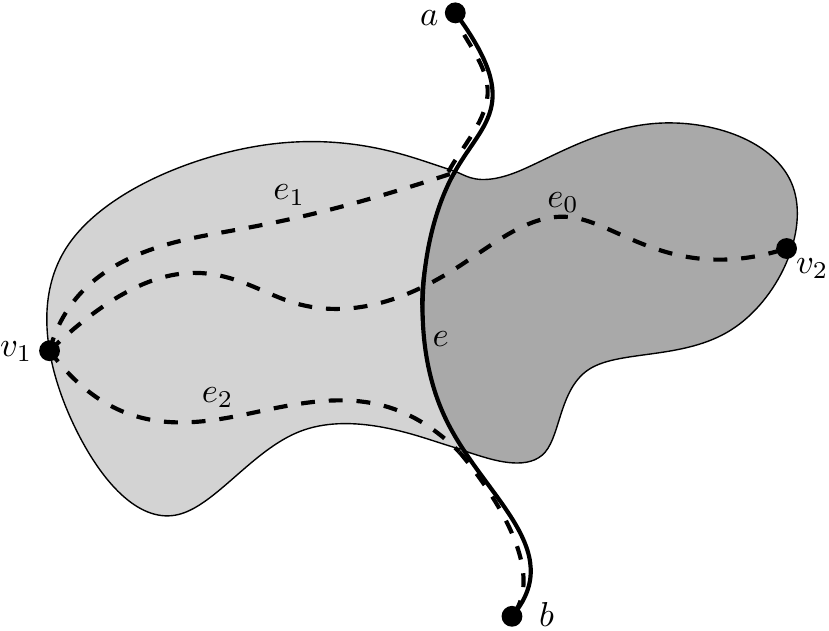}
 \caption{Illustration for the proof of Lemma~\ref{lem:h-con}: Two adjacent faces in the planarization $\Gamma^*$
 that are incident to two distinct vertices, $v_1$ and $v_2$, separated by an edge $e=(a,b)$.
 If we replace the edge $e$ by either $e_0=(v_1,v_2)$ or both $e_1=(v_1,a)$ and $e_2=(v_1,b)$,
 we obtain a \oneblip drawing of a graph that has either one fewer component in $H$ or one more edge.}\label{fig:hconec}
\end{figure}

Recall that $H=(V,E')$ is a sub-multigraph of $G$, where $E'\subseteq E$ is a maximum multiset of edges that are pairwise noncrossing in $\Gamma[E']$, and if there are several such sub-mustligraphs, then $H$ has the fewest connected components.

\begin{lemma}\label{lem:h-con}
 Graph $H=(V,E')$ is connected.
\end{lemma}
\begin{proof}
 Suppose, to the contrary, that $H$ is disconnected. Let $H_1=(V_1,E_1')$ be one component, and let $H_2=(V_2,E_2')$, where $V_2=V\setminus V_1$ and $E_2'=E'\setminus E_1'$.

 Consider the faces in the planarization $\Gamma^*$ of $\Gamma$. Notice that there is no face in $\Gamma^*$ incident to a vertex $v_1\in V_1$ and a vertex $v_2\in V_2$, otherwise we could either add a new edge $(v_1,v_2)$ contradicting the maximality of $G$, or redraw an existing  edge $(v_1,v_2)$ to pass through the interior of this face, contradicting the maximality of $E'$.

 Consequently, we can partition the faces in $\Gamma^*$ into three categories: For $i=1,2$, let $F_i$ be the set of faces incident to a vertex in $V_i$; and let $F_0$ be the set of faces incident to neither $V_1$ nor $V_2$. By Lemma~\ref{lem:cycle}, the region obtained by removing all faces in $F_0$ (i.e., $\mathbb{R}^2\setminus \bigcup_{f\in F_0}$) is connected. Consequently, there exist some faces $f_1\in F_1$ and $f_2\in F_2$ that have a common edge in $\Gamma^*$. Let $v_1\in V_1$ and $v_2\in V_2$ be incident to $f_1\in F_1$ and $f_2\in F_2$. Let $e\in E$ be the edge on the common boundary of $f_1$ and $f_2$, and denote its endpoints by $a,b\in V$.

 We consider three possible edges that we describe together with their drawings (up to homotopy equivalence) with respect to $\Gamma$: Let $e_0=(v_1,v_2)$ be an edge such that it lies in $f_1\cup f_2$; let $e_1=(v_1,a)$ (resp., $e_2=(v_1,b)$) be an edge such that it starts in $f_1$ and closely follows edge $e$ from $f_1$ to its endpoint $a$ (resp., $b$). Refer to Fig.~\ref{fig:hconec}. (If edge $e_0$ (resp., $e_1$ or $e_2$) is homotopic to an existing edge in $\Gamma$, then we can redraw it as described above, and maintain a \oneblip drawing of $G$).

 Note that $e_0\in E$, otherwise we can add $e_0$ to $E$ with the drawing described above, and charge the crossing $e_0\cap e$ to $e_0$, contradicting the maximality of $G$. Note also that $e_1$ and $e_2$ (which may or may not be present in $G$) form a path between $a$ and $b$.  We distinguish two cases:

 \begin{itemize}
  \item Assume $e\notin E'$. We can add $e_0$ to $E'$, contradicting the maximality of $E'$.
  \item Assume $e\in E'$. If neither $e_1$ nor $e_2$ is present in $G$ and $\Gamma$, then we can modify $E$ by replacing $e$ with these edges, contradicting the maximality of $E$. If both $e_1$ and $e_2$ are present in $G$, then they both are in $E'$ by the maximality of $E'$. In this case, we can modify $E'$ by replacing $e$ with $e_0$. Then $a$, $b$, $v_1$, and $v_2$ will be in the same component of $H$, contradicting the tie-breaking rule that $H$ was a maximum crossing-free subgraph with the fewest components. Otherwise we can modify both $E$ and $E'$ by replacing $e$ with $e_1$ or $e_2$ (whichever is not already present in $\Gamma$), and then add edge $e_0$ to $E'$, which contradicts the maximality of $H$.
 \end{itemize}
 All cases lead to a contradiction, which completes the proof.
\end{proof}

In the proof of Lemma~\ref{lem:H}, we shall use Sperner's Lemma~\cite{Sperner28}, a well-known discrete analogue of Brouwer's fixed point theorem.
\begin{lemma} {\rm (Sperner~\cite{Sperner28})} Let $K$ be a geometric simplicial complex in the plane, where the union of faces is homeomorphic to a disk. Assume that each vertex is assigned a color from the set $\{1,2,3\}$ such that three vertices $v_1,v_2,v_3\in \partial K$ are colored 1, 2, and 3, respectively, and for any pair $i,j\in \{1,2,3\}$, the vertices on the path between $v_i$ and $v_j$ along $\partial K$ that does not contain the 3rd vertex are colored with $\{i,j\}$. Then $K$ contains a triangle whose vertices have all three different colors.
\end{lemma}

We are now ready to prove Lemma~\ref{lem:H}, restated in the following form.

\ShoLong{
}{
 \setcounter{lemma}{0}
}

\begin{lemma}\label{lem:htriang}
 The multigraph $H$ is a triangulation, that is, a plane multi-graph in which every face is bounded by a walk with three vertices and three edges.
\end{lemma}
\begin{proof}
 Suppose, to the contrary, that $H$ is not a triangulation. Then $H$ has a face $f$ whose boundary walk $w=(v_1,v_2,\ldots, v_m)$ has more than three vertices (i.e., $m\geq 4$). To simplify notation, we assume that $w$ is a simple cycle; this assumption is not essential for the proof.

 Let $P_f$ be the subgraph of $\Gamma^*$ formed by all edges and vertices lying in the interior or on the boundary of $f$; let $V_f$ denote the set of vertices of $P_f$ (it consists of $v_1,\ldots, v_m$ and all crossings in the interior or on the boundary of $f$); and let $F$ denote the set of faces of $\Gamma^*$ that lie in $f$.
 Let $F_0\subseteq F$ be the set of faces that are not incident to any vertex in $\{v_1, \ldots, v_m\}$; and for $i=1,\ldots , m$, let $F_i\subseteq F$ be the set of faces incident to $v_i$.

 We note the following properties of the arrangement of faces in $F$.
  \begin{itemize}
  \item[{\rm (P1)}] A face $f_i\in F_i$ cannot be incident to a vertex $v_j$, $j\notin \{i-1,i,i+1\}$. Indeed, otherwise we could add a new edge $e=(v_i,v_j)$ to $G$ that lies in $f_i$. Note that $\Gamma$ does not contain a homotopic parallel edge, otherwise it would lie in the face $f$, and could be added to $H$, contradicting the maximality of $H$.
  \item[{\rm (P2)}] A face $f_i\in F_i$ cannot be adjacent to a face $f_j\in F_j$, $j\notin \{i-1,i,i+1\}$.
        Indeed, otherwise we can add a new edge $(v_i,v_j)$ to $G$ such that $(v_i,v_j)$ lies in $f_i\cup f_j$ and uses a gap to cross the boundary between these faces (Fig.~\ref{fig:sperner}(a)).
        Again $\Gamma$ cannot contain a homotopic parallel edge, otherwise it would lie in the face $f$, and could be added to $H$, contradicting the maximality of $H$.
  \item[{\rm (P3)}] A vertex $c\in V_f\setminus V$ cannot is incident to two faces $f_i\in F_i$ and $f_j\in F_j$, $j\notin \{i-1,i,i+1\}$ Suppose, to the contrary, that there is such a vertex $c$ (Fig.~\ref{fig:sperner}(b)). Then two edges $e_1,e_2\in E\setminus E'$ cross at $c$. We can replace edge $e_1$ with a new edge $(v_i,v_j)$ that lies in $f_i\cup f_j$ and that crosses edge $e_2$ at $c$. The new edge can be inserted into both $G$ and $H$, contradicting the maximality of $H$. In this case, $\Gamma$ cannot already contain a homotopic parallel edge,  otherwise it could be added to $H$, contradicting the maximality of $H$.
  \item[{\rm (P4)}] A face $f_0\in F_0$ cannot be adjacent to two faces $f_i\in F_i$ and $f_j\in F_j$  $j\notin \{i-1,i,i+1\}$. Suppose to the contrary that there is such a face $f_0$ (Fig.~\ref{fig:sperner}(c)).
        Then two edges $e_1,e_2\in E\setminus E'$ are on the common boundary of the adjacent pairs $f_i,f_0$ and $f_0,f_j$. We can replace edge $e_1$ with a new edge $(v_i,v_j)$ that lies in $f_i\cup f_0\cup f_j$ that crosses edge $e_2$. The new edge can be inserted into both $G$ and $H$, contradicting the maximality of $H$. Again, $\Gamma$ cannot already contain a homotopic parallel edge, otherwise it could be added to $H$, contradicting the maximality of $H$.
 \end{itemize}

 We next distinguish two cases.

 \medskip\noindent {\bf Case 1. For every $i\in \{1,\ldots , m\}$, the edge $(v_i,v_{i+1})$ is incident to faces in $F_0\cup F_i\cup F_{i+1}$ only.}
 We use Sperner's Lemma~\cite{Sperner28} for a triangulation $K$ of the \emph{dual graph} on the faces
 $F_1\cup \cdots\cup F_m$, that we define here. We first create the \emph{standard dual graph} of $F$: The nodes correspond to the faces in $F$; and two nodes are adjacent if and only if the corresponding faces are adjacent in $\Gamma^*$. We then triangulate the standard dual graph as follows. For every crossing $c\in V_f$ in the interior of $f$ is incident to four faces in $F$, and their adjacency graph forms a 4-cycle in the standard dual. By Lemma~\ref{lem:cycle}(2), at least three of those faces are in $F\setminus F_0$. We triangulate the 4-cycle by an arbitrary diagonal between two faces in $F\setminus F_0$. Note that the faces in $F_0$ still form an independent set by Lemma~\ref{lem:cycle}(2). We call the resulting graph the \emph{modified dual graph} of $F$. By Property (P4), every face in $F_0$ is adjacent to at most one side of $f$. Consequently, the modified dual graph is 2-connected, and the neighbors each face $f_0\in F_0$ form a cycle or a path. Finally, remove all nodes corresponding to $F_0$ from the modified dual graph, and triangulate the cycle or path of neighboring nodes arbitrarily to obtain a triangulation $K$. The condition in Case~1 implies that $K$ is a geometric simplicial complex, where the union of faces is homeomorphic to a disk.

 We now define a 3-coloring of $K$ (the coloring need not be proper). Assign color 1 to all faces in $F_1$.  For $i=2,\ldots , m$, assign color 2 to all faces in $F_i\setminus \bigcup_{j<i}F_j$ if $i$ is even, and color 3 if $i$ is odd. Since $m\geq 4$, each of the three colors are used at least once.

\begin{figure}[htbp]
 \centering
 \includegraphics[width=0.9\columnwidth]{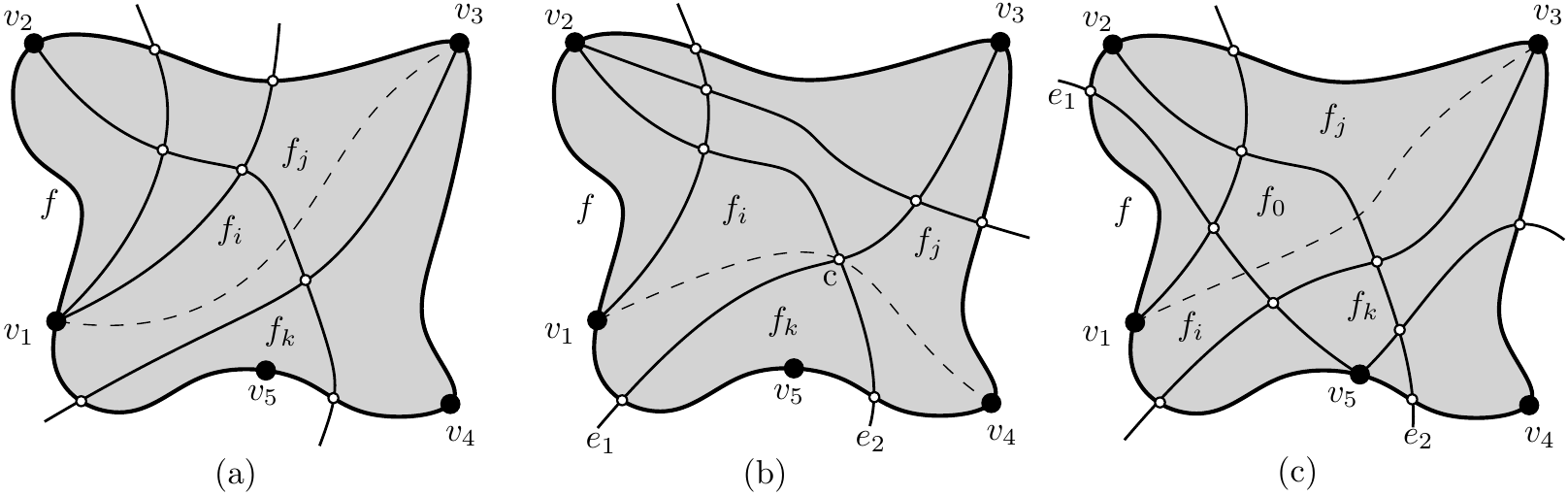}
 \caption{Illustration for the proof of Lemma~\ref{lem:htriang} with $m=5$:
 The dual graph $K$ (which is not shown in the figure) has a triangle whose nodes have
 three different colors, say $f_i\in F_i$, $f_j\in F_j$, and $f_k\in F_k$,
 where $j\notin \{i-1,i+1\}$.
 (a) Faces $f_i$ and $f_j$ are adjacent.
 (b) Vertex $c\in V_f$ is incident to both $f_i$ and $f_j$.
 (c) A face $f_0\in F_0$ is adjacent to both $f_i$ and $f_j$.}\label{fig:sperner}
\end{figure}

 We have seen that $K$ satisfies the conditions of Sperner's Lemma. The Lemma implies that $K$ contains a triangle whose nodes have all three different colors, say $f_i\in F_i$, $f_j\in F_j$, and $f_k\in F_k$. Without loss of generality, assume that $j\notin \{i-1,i+1\}$. (Possibly we have $k\in\{i-1,i+1\}\cap \{j-1,j+1\}$, e.g., $i=1$, $k=2$, and $j=3$.) Consider three cases depending on how the edge $(f_i,f_j)$ in $K$ was created:
 \begin{itemize}
 \item If $f_i$ and $f_j$ are adjacent in $\Gamma^*$, then (P2) is violated.
 \item If a vertex $c\in V_f$ is incident to both $f_i$ and $f_j$, then (P3) is violated.
 \item If a face $f_0\in F_0$ is adjacent to both $f_i$ and $f_j$, then (P4) is violated.
 \end{itemize}
 All three subcases lead to a contradiction.

\medskip\noindent {\bf Case 2. There is an index $i\in \{1,\ldots , m\}$ such that $(v_i,v_{i+1})$ is incident to a face in $F_j$ for some $j\neq 0,i,i+1$.}
 Without loss of generality, we may assume that edge $(v_1,v_m)$ is incident to a face in $F_j$ for some $1<j<m$.
  (Refer to Fig.~\ref{fig:flat} where $m=5$.)
 Note that edge $(v_1,v_m)$ must be incident to some face in $F_j$ for \emph{all} $1\leq j \leq m$;
 otherwise $(v_1,v_m)$ would be incident to two faces, $f_i\in F_i$ and $f_j\in F_j$, $j\notin \{i-1,i,i+1\}$, that are either adjacent to each other or both adjacent to some face $f_0\in F_0$; and then we could add a new edge $(v_i,v_j)$ lying in $f_i\cup f_j$ or $f_i\cup f_0\cup f_j$.

 It follows that there are faces $f_2\in F_2$ and $f_3\in F_3$ that are incident to some point $c\in (v_1,v_m)$ (see Fig.~\ref{fig:flat}(a)); or both are adjacent to some common face $f_0\in F_0$ that is incident to $(v_1,v_m)$ (see Fig.~\ref{fig:flat}(b)).

\begin{figure}[htbp]
 \centering
 \includegraphics[width=0.95\columnwidth]{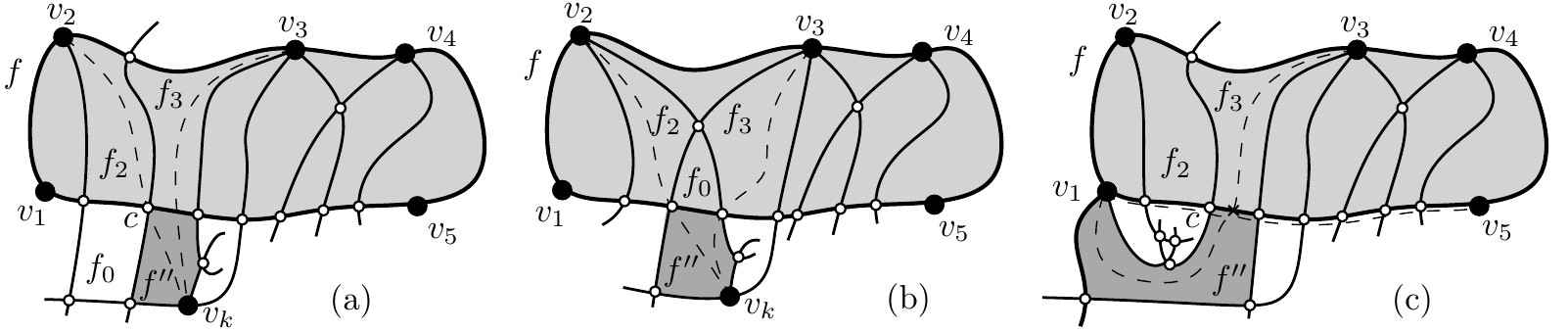}
 \caption{Illustration for the proof of Lemma~\ref{lem:htriang} with $m=5$:
 Edge $(v_1,v_5)$ is incident to some face in $F_j$ for \emph{all} $1\leq j \leq 5$.
 (a) Faces $f_2\in F_2$ and $f_3\in F_3$ that are incident to some point $c\in (v_1,v_m)$.
 (b) Faces $f_2\in F_2$ and $f_3\in F_3$ adjacent to face $f_0\in F_0$ that is incident to $(v_1,v_m)$.
 (c) Face $f''$ is incident to $v_1$.}\label{fig:flat}
\end{figure}

 Consider the face $f'$ of $H$ on the opposite side of $(v_1,v_m)$, and let $F'$ be the set of faces in the planarization $\Gamma^*$ contained in $f'$. Let $f''\in F'$ be a face incident to $c\in (v_1,v_m)$ or adjacent to face $f_0$. By Lemma~\ref{lem:cycle}(2), we may assume that $f''$ is incident to a vertex $v_k$ on the boundary of the face $f'$. It is possible that $v_k=v_1$ or $v_k=v_m$.

\begin{itemize}
 \item If $v_k\notin\{v_1,v_m\}$, we modify $G$, $\Gamma$, and $H$ as follows (Fig.~\ref{fig:flat}(a)--(b)):
 Consider the possible edges $(v_2,v_k)$ and $(v_3,v_k)$ that lie in $f_2\cup f''$ and $f_3\cup f''$, respectively, they each cross $(v_1,v_m)$ and at most one additional edge at $c$ or at a vertex of $f_0$. If $(v_2,v_k)$ or $(v_3,v_k)$ is present in $G$ and $\Gamma$ (as a homotopic copy), it can be redrawn to lie in $f_2\cup f''$ and $f_3\cup f''$, respectively. If $(v_2,v_k)$ or $(v_3,v_k)$ is not present in $G$ and $\Gamma$, we then insert it and remove the edge $(v_1,v_m)$.
 Finally, we can modify $E'$ by replacing $(v_1,v_m)$ with $(v_2,v_k)$ and $(v_3,v_k)$, contradicting the maximality of $E'$.
  \item If $v_k=v_1$, then we modify $G$, $\Gamma$, and $H$ as follows (Fig.~\ref{fig:flat}(c)):
  Add a new edge $(v_1,v_3)$ that lies in $f_3\cup f''$ or $f_3\cup f_0\cup f''$, and crosses $(v_1,v_m)$
  at a point $x$ on the boundary between $f''$ and $f_3$. Then redraw the edges $(v_1,v_m)$ and $(v_1,v_3)$ by exchanging their initial arcs between $v_1$ and $x$, and eliminating the crossing at $x$. The edge $(v_1,v_3)$ was not previously present in $G$, otherwise it would be homotopic to a diagonal $(v_1,v_3)$ of the face $f$ of $H$, contradicting the maximality of $E'$.
  (However, a homotopic copy of the new drawing of edge $(v_1,v_m)$ may be already present in $\Gamma$, in which case, the total number of edges in $G$ remains the same). Modify $E'$ by replacing the edge $(v_1,v_m)$ of face $f$ with the new edges $(v_1,v_3)$ and $(v_1,v_m)$ described here. This contradicts the maximality of $E'$.
  \item If $v_k=v_m$ and $v_{m-1}=v_3$, we make similar changes: We increase $|E'|$ by replacing the edge $(v_1,v_m)$ of $f$ with a new edge $(v_2,v_m)$ and a new drawing of the edge $(v_1,v_m)$.
 \end{itemize}
 All cases lead to a contradiction. Therefore, our initial assumption must be dropped, consequently the multigraph $H$ is a triangulation, as claimed.
\end{proof}
}

\ShoLong{
 \proofLemmaOneSection
}{
}

\ShoLong{
 \subsection{Lower bound constructions}\label{sse:lower-density}
}{
}

We now show that the bound of Theorem~\ref{thm:density} is worst-case optimal.
A 2-planar graph with $n$ vertices and $5n-10$ edges is also \oneblip (see Lemma~\ref{le:2kplanarekblip}). Pach and T\'oth~\cite{DBLP:journals/combinatorica/PachT97} construct such a graph by starting with a plane graph with pentagonal faces (e.g., using nested copies of an icosahedron), and then add all five diagonals in each pentagonal face; see Fig.~\ref{fig:5n-a}.
This construction yields a \oneblip graph with $n$ vertices and $m=5n-10$ edges for all $n\geq 20$, $n\equiv 5 \pmod{15}$.

\begin{figure}[tb]
 \centering
 \subfigure[]{\includegraphics[page=1]{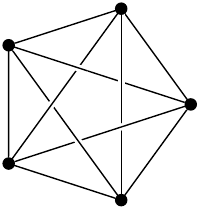}\label{fig:5n-a}}\hfil
 \subfigure[]{\includegraphics[page=2]{figs/5n-new}\label{fig:5n-b}}\hfil
 \subfigure[]{\includegraphics[page=3]{figs/5n-new}\label{fig:5n-c}}
 \ShoLong{
  \hfil\subfigure[]{\includegraphics[page=3]{figs/5n}\label{fig:5n-d}}
 }{
 }
 \caption{Patterns that produce \oneblip graphs with $n$ vertices and $5n-\Theta(1)$ edges.}
\end{figure}
We can modify this construction by inserting a new vertex in one or more pentagons, and connecting it to the 5 vertices of the pentagon; see Fig.~\ref{fig:5n-b}. Every new edge crosses exactly one diagonal of the pentagon, so the new crossings can be charged to the new edges. Since every new vertex has degree 5, the equation $m=5n-10$ prevails. By inserting a suitable number of vertices into pentagons, we obtain constructions for $n\in\mathbb{N}$ such that $20\leq n\leq 32$ or $n\geq 38$.
A similar construction is based on hexagonal faces; see Fig.~\ref{fig:5n-c}. Start with a \emph{fullerene}, that is, a 3-regular, plane graph $G_0$ with $n_0$ vertices, 12 pentagon faces, and $n_0/2-10$ hexagon faces (including the external face).
Add diagonals in each face to connect a vertex to their second neighbors (the graph is 2-planar so far); finally insert a new vertex in each face of $G_0$, and connect them to all vertices of that face. We obtain a \oneblip graph $G$. The number of vertices is $n=n_0+12+(n_0/2-10)=\frac{3}{2}n_0+2$, and the number of edges is $m=\frac{3}{2}n_0+10\cdot 12+ 12\cdot (n_0/2-10)= \frac{15}{2}n_0= 5n-10$.
Fullerenes exist for $n_0=20$ and for all even integers $n_0\geq 24$~\cite{DBLP:journals/jcisd/BrinkmannGM12}. This yields a lower bound of $5n-10$ for $n=32$ and for all $n\geq 35$ where $n\equiv 2\mod 3$.
However, similarly to the previous construction, the equation $m=5n-10$ prevails if we \emph{delete}
up to 12 vertices inserted into pentagons. Consequently, the upper bound $5n-10$ in Theorem~\ref{thm:density} is tight for all $n\geq 20$.

\begin{theorem}
 For every integer $n \ge 20$ there exists a \oneblip (simple) graph with $n$ vertices and $5n-10$ edges.
\end{theorem}

\newcommand{\mentionLowerBounds}{
 We mention a third, slightly weaker construction, which is based on a sequence of nested squares. Fig.~\ref{fig:5n-d} shows how to add 16 edges between two consecutive squares such that the 16 crossings are assigned to distinct edges. We can add two diagonals in the external face and the innermost square. Using $s$ squares, we have $n=4s$, and $m=4s+16(s-1)+2\cdot 2=20s-12=5n-12$.  In particular, for $s=2$ this yields a drawing of $K_{8}$; see Fig.~\ref{fig:k8}.
}
\ShoLong{
 \mentionLowerBounds
}{
}

If we allow \oneblip muligraphs (with nonhomotopic parallel edges in a \oneblip drawing), then we can construct smaller configurations for which the upper bound $5n-10$ of Theorem~\ref{thm:density} is tight. Start with a regular polygon $P_0$ with $n_0\geq 5$ vertices. Subdivide the interior and the exterior of $P_0$ independently into one pentagon and $n_0-5$ triangle faces using $n_0-5$ diagonals. In each of the two pentagons, add five edges as shown in Fig.~\ref{fig:charge}(left). In each triangle, add a new vertex and six new edges as shown in Fig.~\ref{fig:charge}(right).
We obtain a \oneblip drawing of a multigraph with $n=n_0+2(n_0-5)=3n_0-10$ vertices and $n_0+ 2(n_0-5)+2\cdot 5+2(n_0-5)\cdot 6=15n_0-60=5n-10$ edges for all $n\geq 5$, $n\equiv 2\mod 3$. By inserting a new vertex in one or two pentagons, and connecting it to the 5 vertices of the pentagon as in Fig.~\ref{fig:5n-b}, the lower bound extends for all integers $n\geq 5$. We summarize our lower bound for multigraphs in the following theorem.

\begin{theorem}
 For every integer $n \ge 5$, there exists a \oneblip multigraph with $n$ vertices and $5n-10$ edges.
\end{theorem}


\section{Relationship between \kblip graphs and other families of beyond-planar graphs}
\label{sec:relationship}
In this section we prove the following theorem.

\begin{theorem}\label{thm:relationships}
 For every integer $k \ge 1$, the following relationships hold.
 $$ (2k)\mbox{\sc -planar}  \subsetneq k\mbox{\sc -gap-planar} \subsetneq (2k+2)\mbox{\sc -quasiplanar}$$
\end{theorem}

\noindent We begin by showing the following.

\begin{lemma}\label{le:quasiplanar}
 For all $k \ge 1$, every \kblip drawing is $(2k+2)$-quasiplanar.
\end{lemma}

\noindent We also need to show that for every $k\in \mathbb{N}$ there is
a $(2k+2)$-quasiplanar graph that is not \kblip. We prove a stronger statement:

\begin{lemma}\label{le:threequasiplanar}
 For all $k \ge 1$, there is a 3-quasiplanar graph $G_k$ that is not \kblip.
\end{lemma}
\begin{proof}
 Let $k\in \mathbb{N}$. We construct a graph $G_k=(V,E)$ as follows.
 Start with $K_{3,3}$ and replace each edge by $t=19k$ edge-disjoint paths of length $2$.
 Note that the total number of edges is $|E|=9\cdot 2t=18t$.
 Graph $G_k$ is $3$-quasiplanar. Since ${\rm cr}(K_{3,3})=1$, it admits
 a drawing with precisely one crossing. The paths of length 2 can be drawn close
 to the edges of $K_{3,3}$ such that two paths cross if and only if the two
 corresponding edges of $K_{3,3}$ cross. Consequently, $G_k$ admits a drawing $\Gamma_0$
 in which any two crossing edges are part of two paths that correspond to two
 crossing edges of $K_{3,3}$, which in turn implies that no three edges in $\Gamma_0$
 pairwise cross.

 Suppose that $G_k$ admits a \kblip drawing $\Gamma$.
 Then the total number of crossings is at most $k|E|=18kt$.
 We derive a contradiction by showing that ${\rm cr}(G_k)\geq 19kt$.
 If we choose one of the $t$ paths for each of the $9$ edges of $K_{3,3}$ independently,
 then we obtain a subdivision of $K_{3,3}$, therefore there is a crossing between at least
 one pair of paths. There are $t^9$ ways to choose a path for each of the 9 edges
 of $K_{3,3}$. Each crossing between two paths in $\Gamma$ is counted $t^{9-2}=t^7$ times.
 Consequently, the total number of crossings in $\Gamma$ is at least $t^2=19kt$.
\end{proof}

We now show that every $(2k)$-planar drawing is \kblip. We note that a similar result can be  derived from~\cite{DBLP:conf/fun/BruckdorferK12} (Lemma 10), but only for the case $k=1$. A bipartite graph with vertex sets $A$ and $B$ is denoted as $H=(A,B,E)$.  A \emph{matching} from $A$ into $B$ is a set $M \subseteq E$ such that every vertex in $A$ is incident to exactly one edge in $M$ and every vertex in $B$ is incident to at most one edge in $M$. The \emph{neighborhood} of a subset $A' \subseteq A$ is the set of all vertices in $B$ that are adjacent to a vertex in $A'$, and is denoted as $N(A')$.

\begin{lemma}\label{le:2kplanarekblip}
 For all $k \ge 1$, every $(2k)$-planar drawing is \kblip.
\end{lemma}
\begin{proof}
 let $k\in \mathbb{N}$, $k\geq 1$, let $G$ be a $(2k)$-planar graph, and let $\Gamma$ be a $(2k)$-planar drawing of $G$. Let $H = (A \cup B,E_H)$ be a bipartite graph obtained as follows. The set $A$ has a vertex $a_{e,f}$ for each crossing in $\Gamma$ between two edges $e$ and $f$ of $G$. For each edge $e$ of $G$ there are $k$ vertices $b^1_e,\dots,b^{k}_e$ in $B$.  For every pair of edges $e, f$ of $G$ that cross in $\Gamma$, graph $H$ contains edges  $(a_{e,f},b^1_e),\dots,(a_{e,f},b^{k}_e)$ and $(a_{e,f},b^1_{f}),\dots,(a_{e,f},b^{k}_{f})$ in $H$. Notice that if $H$ admits a matching of $A$ in $B$, then each crossing of $\Gamma$ between an edge $e$ and an edge $f$ of $G$ can be assigned to either $e$ or $f$, and no edge of $G$ is assigned with more than $k$ crossings.

 Consider any subset $A'$ of $A$, and let $B' = N(A')$ be the neighborhood of $A'$ in $B$. We claim that $|A'| \le |B'|$.  Let $E' \subseteq E_H$ denote the edges between $A'$ and $B'$.  By construction every vertex in $A$ has degree $2k$, and hence $|E'| \ge 2k |A'|$.  On the other hand, every vertex in $B$ has degree at most $2k$ as every edge of $G$ has at most $2k$ crossings, and hence $|E'| \le 2k |B'|$.  Hence $|A'| \le |B'|$ as claimed.

 By Hall's theorem, it now follows that $H$ admits a matching from $A$ into $B$, which corresponds to an assignment of gaps in $\Gamma$ such that no edge has more than $k$ gaps, i.e., $\Gamma$ is a \kblip drawing.
\end{proof}

To conclude the proof of Theorem~\ref{thm:relationships}, we should prove that for every $k \ge 1$, there is a \kblip graph that is not $(2k)$-planar. A stronger result holds:

\begin{lemma}\label{le:oneblipnotkplanar}
 For every $k \ge 1$, there exists a \oneblip graph $G_k$ that is not $k$-planar.
\end{lemma}
\newcommand{\proofoneblipnotkplanar}{
 \begin{proof}
  Let $k\in \mathbb{N}$. We construct a graph $G_k=(V,E)$ together with its \oneblip drawing as follows.
  Start with an edge $(a,b)$ crossed by $k+1$ disjoint edges $(c_i,d_i)$, for $i=1,\ldots , k+1$.
  The $2k+2$ vertices lie in a common face, and we can connect them by a Jordan curve,
  which forms a cycle $C=(a,c_1,\ldots , c_{k+1},b,d_{k+1},\ldots , d_1)$. Add a new vertex
  $v_0$ in the exterior of the cycle, and connect it to all vertices of $C$. The cycle $C$ and
  $v_0$ form the wheel $W$, which has $m=4k+4$ edges.
  Finally, replace each edge of the wheel by $t$ edge-disjoint paths of length 2,
  where $t\geq k$ is a suitable parameter \ShoLong{that we shall specify shortly}{(details on the exact choice can be found in the extended version of the paper)}.
  This completes the construction of $G_k=(V,E)$.
  Note that the total number of edges is bounded above by
  $$|E|=1+(k+1)+ (4k+4)2t=1+(k+1)(8t+1)< 10(k+1)t.$$

  It is clear that $G_k$ is \oneblip, since the crossing between $(a,b)$
  and $(c_i,d_i)$ can be charged to $(c_i,d_i)$ for all $i=1,\ldots, k+1$.

  Suppose that $G_k$ admits a $k$-planar drawing $\Gamma$.
  Since each edge crosses at most $k$ other edges,
  the total number of crossings is at most $k|E|/2<5k(k+1)t$.

  We claim that for each edge of the wheel $W$, we can choose $k+1$ of the $t$ paths such
  that no two chosen paths that correspond to different edges of the wheel cross in the drawing $\Gamma$.
  We prove the claim by contradiction. Since we choose $k+1$ out of $t$ paths for each of the $m$ edges
  of the wheel independently, there are ${t\choose k+1}^m$ possible choices.
  Suppose, for the sake of contradiction, that every choice produces a graph that
  has at least one crossing in $\Gamma$ between paths corresponding to different edges of $W$.
  Each crossing between two such paths is counted ${t-1 \choose k}^2 {t\choose k+1}^{m-2}$ times.
  Consequently, the total number of crossings in $\Gamma$ is at least $t^2/(k+1)^2$.
  If we put $t=5(k+1)^4$, then we would have at least $t^2/(k+1)^2=t\cdot 5(k+1)^4/(k+1)2=5(k+1)^2t>5k(k+1)t$ crossings,
  a contradiction. This completes the proof of the claim.

  Let $G_k'$ be a subgraph of $G_k$ that consists of $k+1$ paths corresponding to each edge of $W$
  such that the paths corresponding to different edges of $W$ do not cross in $\Gamma$; and let $\Gamma'$ be
  the restriction of $\Gamma$ to $G_k'$. Note that any $k+1$ paths that correspond to the same
  edge of $W$ are homotopic to each other in $\Gamma'$. If we pick one of the $k+1$ paths, for each edge
  of $W$, the Jordan arc along these paths provide a planar drawing of $W$.
  Since $W$ is 3-connected, it has a combinatorially unique embedding, which we denote by $\Gamma(W)$.
  As noted above, every edge of $W$ in the drawing $\Gamma(W)$ is homotopic to $k+1$ paths in the drawing $\Gamma'$.

  As the combinatorial embedding of $W$ is unique, the vertices $a,b$, and $c_i,d_i$, for $i=1,\ldots ,k+1$, lie on the boundary of a single face, which we denote by $F$. If edges $(a,b)$ and $(c_i,d_i)$, for $i=1,\ldots , k+1$, are homotopic to Jordan arcs that lie in $F$, then $(a,b)$ crosses $(c_i,d_i)$, for all $i=1,\ldots , k+1$.
  If any of these edges is not homotopic to a Jordan arc in $F$, then it crosses a bundle of $k+1$ paths
  corresponding to some edge of $W$. In both cases, one of the edges crosses $k+1$ other edges in $\Gamma$,
  contradicting our assumption that $\Gamma$ is a $k$-planar drawing.
 \end{proof}
}
\ShoLong{
 \proofoneblipnotkplanar
}{
}

\paragraph{Relation to $d$-degenerate crossing graphs.}
Eppstein and Gupta~\cite{eg-cpnrn-17}  defined $d$-degenerate crossing graphs, for $d\in \mathbb{N}$. A graph is a \emph{$d$-degenerate crossing graph} if it admits a drawing $\Gamma$ such that the crossing graph $C(\Gamma)$ is $d$-degenerate. Recall that a graph is \emph{$d$-degenerate} if the vertices admit a total order in which each vertex is adjacent to at most $d$ previous vertices. It is clear from the definition that for every $k\in \mathbb{N}$, every $k$-degenerate crossing graph is a \kblip graph. However, the converse is false for $k=1$: We show below (Lemma~\ref{le:1degenerate}) that for every \oneblip drawing of a \oneblip graph with $n\geq 20$ vertices and the maximum number of edges, the crossing graph contains a cycle, hence it is not 1-degenerate.

\begin{lemma}\label{le:1degenerate}
For every \oneblip graph $G$ with $n\geq 20$ vertices and $5n-10$ edges
and for every \oneblip drawing $\Gamma$ of $G$,
the crossing graph $C(\Gamma)$ contains a cycle.
Consequently, $C(\Gamma)$ is not 1-degenerate, and $G$ is not a 1-degenerate crossing graph.
\end{lemma}
\begin{proof}
Let $G=(V,E)$ be a \oneblip graph with $n\geq 20$ vertices and $5n-10$ edges
(infinite families of such graphs have been constructed in Section~\ref{sse:lower-density}).
Let $\Gamma$ be a \oneblip drawing of $G$, and let $C(\Gamma)=(E,X)$ be its crossing graph.

Let $H=(V,E')$ be a subgraph of $G$, where $E'\subseteq E$ is a maximum set of edges that are pairwise noncrossing in $\Gamma$.
By Lemma~\ref{lem:H}, $H$ is a triangulation, consequently $|E'|=3n-6$. Let $E''=E\setminus E'$. The charging scheme in the proof of
Theorem~\ref{thm:density} gives a one-to-one correspondence between $E''$ and the $2n-4$ faces of $H$ such that each edge $e\in E''$
corresponds to an end triangle of $e$. Since $G$ is a simple graph, the two end portions of every edge $e\in E''$ lie in
two different (triangular) faces of $H$.

Starting with an arbitrary edge $e_1\in E''$, we construct a sequence $P$ of edges in $E''$ as follows.
Assume that the edge $e_i$ is already defined for $i\in \mathbb{N}$. Edge $e_i$ has two distinct end triangles,
and the charging scheme matches $e_i$ to only one of them. Let $\Delta_i$ be the end triangle of $e_i$ that
is not charged by $e_i$, and let $e_{i+1}\in E''$ be the edge charged to $\Delta_i$.
Since $E''$ is finite, the sequence $P$ contains a cyclic sequence without repetition that we denote by $C_0$.

We construct a cyclic sequence $C_1$ from $C_0$ that forms a cycle in the crossing graph $C(\Gamma)$ as follows.
Consider two consecutive elements of $C_0$, say $e_i$ and $e_{i+1}$, both have an end portion in triangle $\Delta_i$.
For every two consecutive edges in $C_0$, say $e_i$ and $e_{i+1}$, do the following: If the end portions of $e_i$ and $e_{i+1}$ that lie in $\Delta_i$ are incident to the same vertex of $\Delta_i$, then both edges cross the opposite side of $\Delta_i$ that we denote by $f_i$, and we insert edge $f_i\in E'$ into $C_0$ between $e_i$ and $e_{i+1}$.
Otherwise end portions of $e_i$ and $e_{i+1}$ in $\Delta_i$ are incident to two distinct vertices of $\Delta_i$, consequently $e_i$ and $e_{i+1}$ cross in the interior of $\Delta_i$, and we do not insert anything between $e_i$ and $e_{i+1}$. We obtain a cyclic sequence $C_1$ of edges in $E$ such that every two consecutive edges cross in the drawing $\Gamma$. We have shown that $C(\Gamma)$ contains a cycle, as claimed.
\end{proof}

\textbf{Note:} 
David Wood (private communication) has pointed out that a weak version of the converse is also true. Specifically, every \kblip graph $G$ is $2k$-degenerate. This follows from the fact that the crossing graph $C$ of a \kblip drawing of $G$ has average degree at most $2k$. 

\section{\oneblip drawings of complete graphs}
\label{sec:complete}
In this section, we characterize which complete graphs are \oneblip.

\begin{theorem}\label{th:complete}
The complete graph $K_n$ is \oneblip if and only if $n \le 8$.
\end{theorem}
\begin{proof}
  Figure~\ref{fig:k8} shows a \oneblip drawing of $K_8$, and by
  monotonicity the graphs $K_1,\ldots,K_7$ are \oneblip as well.
  We now prove that $K_9$ is not \oneblip, which again by monotonicity
  settles all cases $K_n$ for $n \geq 9$.

  Since $K_9$ has $36$ edges and ${\rm cr}(K_9) = 36$~\cite{Guy1972},
  a \oneblip drawing of $K_9$ can only arise from
  assigning exactly one gap to each edge in a crossing-minimal drawing of
  $K_9$ (cf.~Property~\ref{pr:edges-crossings}). We obtain a
  contradiction by showing that in every crossing-minimal drawing of
  $K_9$ some edge has no crossing at all.

  Let $\Gamma^*$ be the planarization of such a crossing-minimal
  drawing $\Gamma$. Note that $\Gamma^*$ has $n^*=45$ vertices and
  $m^*=108$ edges (since it has $9$ real vertices of degree $8$ and
  $36$ dummy vertices of degree $4$), so by Euler's formula, the
  number of faces of $\Gamma^*$ is $f^* = m^* - n^* + 2 = 108 - 45 + 2
  = 65$. For a real vertex $u$ of $\Gamma^*$, we denote by $F(u)$ the
  set of faces of $\Gamma^*$ that are incident to $u$. We claim that
  $\Gamma^*$ is biconnected and $|F(u)|=8$ for every real vertex $u$
  of $\Gamma^*$.

  Suppose, for a contradiction, that $\Gamma^*$ is not
  biconnected. Then it contains a cut-vertex $c$, which is either a
  dummy or a real vertex.  If $c$ is a dummy vertex, note that it is
  adjacent to exactly two connected components of $\Gamma^* \setminus \{c\}$. Then
  we can reflect the drawing of one of the two components, thereby
  eliminating the crossing at $c$, which contradicts the
  crossing-minimality of $\Gamma$. We now show that no real vertex is
  a cut-vertex in $\Gamma^*$. Every $3$-cycle in $K_9$ forms a simple
  cycle in $\Gamma^*$ (since $\Gamma$ is a simple drawing and thus
  adjacent edges do not cross). On the other hand, any three real
  vertices in $\Gamma^*$ are part of a $3$-cycle in $K_9$, and thus part of a
  simple cycle in $\Gamma^*$. Hence, no real vertex is a cut-vertex in
  $\Gamma^*$. Finally, $|F(u)|=8$ because every real vertex $u$ has
  degree $8$ and $\Gamma^*$ is biconnected.

  It follows that there are real vertices $u,v$ which share a face
  (i.e. $F(u) \cap F(v) \not= \emptyset$), as otherwise there would
  have to be $\sum_u |F(u)| = 9\cdot 8 = 72 > 65 = f^*$ faces. But now
  the edge $(u,v)$ can be redrawn inside this face so that
  this edge cannot have had any crossing to begin with
  since $\Gamma$ was assumed to be crossing-minimal.
\end{proof}

\begin{figure}[t]
    \centering
    \subfigure[]{\includegraphics[width=0.35\columnwidth]{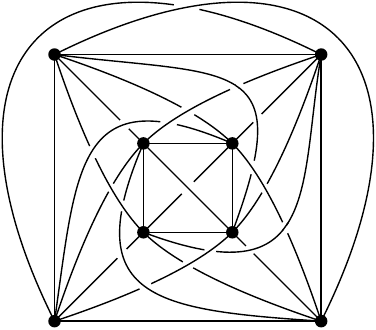}\label{fig:k8}}\hfil
\subfigure[]{\includegraphics[width=0.39\columnwidth]{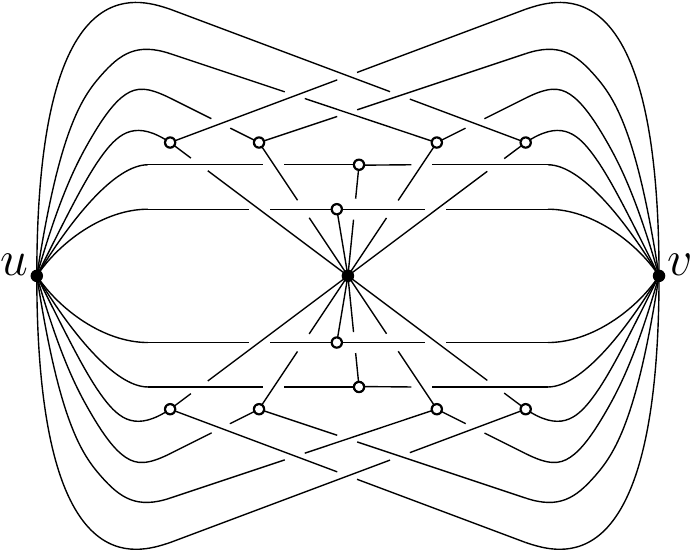}\label{fig:k312}}
  \caption{A \oneblip drawing of (a) $K_8$ and (b) $K_{3,12}$.}
 \end{figure}

\section{Recognizing \oneblip graphs}
\label{sec:recognition}
We denote by \textsc{1GapPlanarity} the problem of deciding whether a given graph $G$ is \oneblip. We show that \textsc{1GapPlanarity} is \textsc{NP}-complete, by a reduction from \textsc{3Partition}.
Recall that an instance  of \textsc{3Partition} consists of a multiset $A=\{a_1,a_2,\ldots ,a_{3m}\}$  of $3m$ positive integers in the range $(I / 4, I / 2)$, where $I$ is an integer such that $I=1/m \cdot \sum^{3m}_{i=1}a_i$, and asks whether $A$ can be partitioned into $m$ subsets  $A_1,A_2,\ldots ,A_m$, each of cardinality 3, such that the sum of integers in each subset is $I$.
This problem is  strongly \textsc{NP}-hard~\cite{DBLP:books/fm/GareyJ79}, and thus we may assume that $I$ is bounded by a polynomial in $m$.

\newcommand{\lemmaNP}{
 \begin{lemma}\label{le:np}
  The problem \textsc{1GapPlanarity} is in \textsc{NP}.
 \end{lemma}
 \begin{proof}
  Given a planarization $\Gamma^{*}$ of a drawing $\Gamma$, we can check whether it is \oneblip in polynomial time by using Property~\ref{pr:pseudoarboricity}. A nondeterministic algorithm to generate all planarizations of a graph with $k$ crossings, where $0 \le k < {{m}\choose{2}}$, evaluates all possible $k$ pairs of edges that cross (and the order of the crossings along the edges) with a technique similar to the one in~\cite{gj-1983}. Then it replaces crossings with dummy vertices and tests whether the resulting graph is planar, i.e., whether it is a planarization of a drawing of $G$, and whether it is \oneblip.  Hence, the problem belongs to \textsc{NP}.
 \end{proof}
}

\noindent \ShoLong{
 We begin by showing that \textsc{1GapPlanarity} is in \textsc{NP}.
 \lemmaNP
}{
 The fact that \textsc{1GapPlanarity} is in \textsc{NP} can easily be shown by exploiting Property~\ref{pr:pseudoarboricity}.
}

\noindent Our reduction is reminiscent to the reduction used in~\cite{Bekos2016}. However, the proof in~\cite{Bekos2016} holds only for the case in which a clockwise order of the edges around each vertex is part of the input, i.e., only if the \emph{rotation system} of the input graph is fixed.  A similar reduction is also used in~\cite{BEKOS201748}, in which the rotation system assumption is not used. However, the gadgets in~\cite{BEKOS201748} have a unique embedding. We do not use the fixed rotation system assumption, nor we can easily derive a unique embedding for our gadgets, and thus have to deal with additional challenges in our proof. In what follows we define a ``blob'' graph that will be used to enforce an ordering among the edges adjacent to certain vertices.
Consider the complete bipartite graph $K_{3,12}$, whose crossing number is  $30$~\cite{KLEITMAN1970315,Zarankiewicz1955}. Fig.~\ref{fig:k312} shows a \oneblip drawing of $K_{3,12}$ with exactly $30$ gaps. Note that two degree-$12$ vertices, $u$ and $v$, are drawn on the outer face. Since $K_{3,12}$ has $36$ edges, the next lemma easily follows.

\begin{lemma}\label{le:k312blipnumber}
 Every \oneblip drawing of $K_{3,12}$ has at most $6$ gap-free edges.
\end{lemma}

A \emph{blob} $B$ is a copy of $K_{3,12}$. A \emph{gapped chain} $\mathcal C$ of a \oneblip drawing is a closed, possibly nonsimple, curve such that any point of $\mathcal C$ either belongs to a gapped edge or corresponds to a vertex.

\begin{lemma}\label{lem:blippedcycle}
 Let $u$ and $v$ be two degree-$12$ vertices~of~$B$. Every \oneblip drawing $\Gamma$ of $B$ contains a gapped chain $\mathcal C$ containing $u$ and $v$.
\end{lemma}
\newcommand{\proofblippedcycle}{
 \begin{proof}
  Let $\Gamma^*$ be the planarization of $\Gamma$. Let $\Gamma'$ be the subgraph of $\Gamma^*$ consisting only of those edges that correspond to or belong to gapped edges of $\Gamma$.  We prove that  $\Gamma'$ contains two edge-disjoint paths from $u$ to $v$.  Note that these two edge-disjoint paths may meet at real vertices and at dummy vertices (i.e., a crossing between two gapped edges).  A curve that goes through these two paths is the desired gapped chain.

  According to Menger's theorem, two such paths exist if and only if every $(u,v)$-cut of $\Gamma'$ has size at least $2$, where a $(u,v)$-cut of $\Gamma'$ is a set of edges of $\Gamma'$ whose removal disconnects $u$ and $v$. It is well known that such an $(u,v)$-cut corresponds to cycle in the dual, which in turn corresponds to a curve that separates $u$ and $v$ by crossing a set of edges. We now consider one such curve, and claim this curve crosses at least two gapped edges in the original drawing $\Gamma$ (after a slight perturbation we can assume that it does not pass through a vertex). More precisely, let $\ell$ be a simple closed curve such that:   it does not pass through any vertex of $\Gamma$;
  it divides the plane into two nonempty topologically connected regions, one containing $u$ and the other containing $v$.  Let $L$ denote the set of edges of $G$ that are crossed by $\ell$.
  Note that $G$ contains $12$ edge-disjoint paths from $u$ to $v$, which induce $12$ edge-disjoint paths from $u$ to $v$ in $\Gamma^*$.  It follows that $|L| \ge 12$. Also, $\Gamma$ has at most $6$ gap-free edges by Lemma~\ref{le:k312blipnumber}. Hence, $L$ contains at least $12-6>2$ gapped edges.\end{proof}
}
\ShoLong{
 \proofblippedcycle
}{
 \begin{sketch}
  Let $\Gamma^*$ be the planarization of $\Gamma$. Let $\Gamma'$ be the subgraph of $\Gamma^*$ consisting only of those edges that correspond to or belong to gapped edges of $\Gamma$.  We prove that  $\Gamma'$ contains two edge-disjoint paths from $u$ to $v$.  Note that these two edge-disjoint paths may meet at real vertices and at dummy vertices (i.e., a crossing between two gapped edges).  A curve that goes through these two paths is the desired gapped chain.
  According to Menger's theorem, two such paths exist if and only if every $(u,v)$-cut of $\Gamma'$ has size at least $2$, where a $(u,v)$-cut of $\Gamma'$ is a set of edges of $\Gamma'$ whose removal disconnects $u$ and $v$. Such edge cuts correspond to cycles in the dual, which in turn correspond to curves that separate $u$ and $v$ by crossing a set of edges. By Lemma~\ref{le:k312blipnumber}, one can show that any such curve crosses at least two gapped edges in the original drawing $\Gamma$.
 \end{sketch}
}

We are now ready to show how to transform an instance $A$ of \textsc{3Partition} into an instance $G_A$ of \textsc{1GapPlanarity}. We start by defining some gadgets for our construction.
A \emph{path gadget} $P_k$ is a graph obtained by merging a sequence of $k>0$ blobs as follows. Let $B_1$, $B_2$, $\dots$, $B_k$, be $k$ blobs such that $u_i$ and $v_i$ are two vertices of degree $12$ in $B_i$. Identify the vertices  $v_i = u_{i+1}$ for $i=1,\dots,k-1$, each of these vertices is called an \emph{attaching vertex}. Thus, $P_k$ has $k+1$ attaching vertices.
In a \oneblip drawing of $P_k$, any two gapped chains of two blobs, $B_i$ and $B_j$ ($i<j$), are disjoint, except for a possible common attaching vertex. A schematization of $P_k$ (for $k=3$) is shown in Fig.~\ref{fig:path-gadget}. A \emph{top beam}, denoted $G_t$, is a path gadget $P_k$ with $k= 3m(\lceil I/2 \rceil+2)+1$. Recall that $G_t$ has $3m(\lceil I/2 \rceil+2)+2$ attaching vertices. A \emph{right wall} $G_r$ is a path gadget $P_k$ with $k= 2$. Symmetrically, a \emph{bottom beam} $G_r$ is a path gadget with $k= 3m(\lceil I/2 \rceil +2)+1$, and a \emph{left wall} $G_l$ is a path gadget with $k=2$. A \emph{global ring} $R$ is obtained by merging $G_t$, $G_r$, $G_b$, and $G_l$ in a cycle as in Fig.~\ref{fig:ring}. Again,  in any \oneblip drawing $\Gamma_R$ of $R$, the gapped chains of two distinct blobs, $B_i$ and $B_j$ ($i\neq j$), are disjoint, except for a possible common attaching vertex. Thus, $\Gamma_R$ contains a gapped chain $C_R$ that is the union of all the  gapped chains of the blobs of $R$.

\begin{figure}[t]
 \centering
 \subfigure[$P_3$]{\includegraphics[width=0.45\columnwidth, page=1]{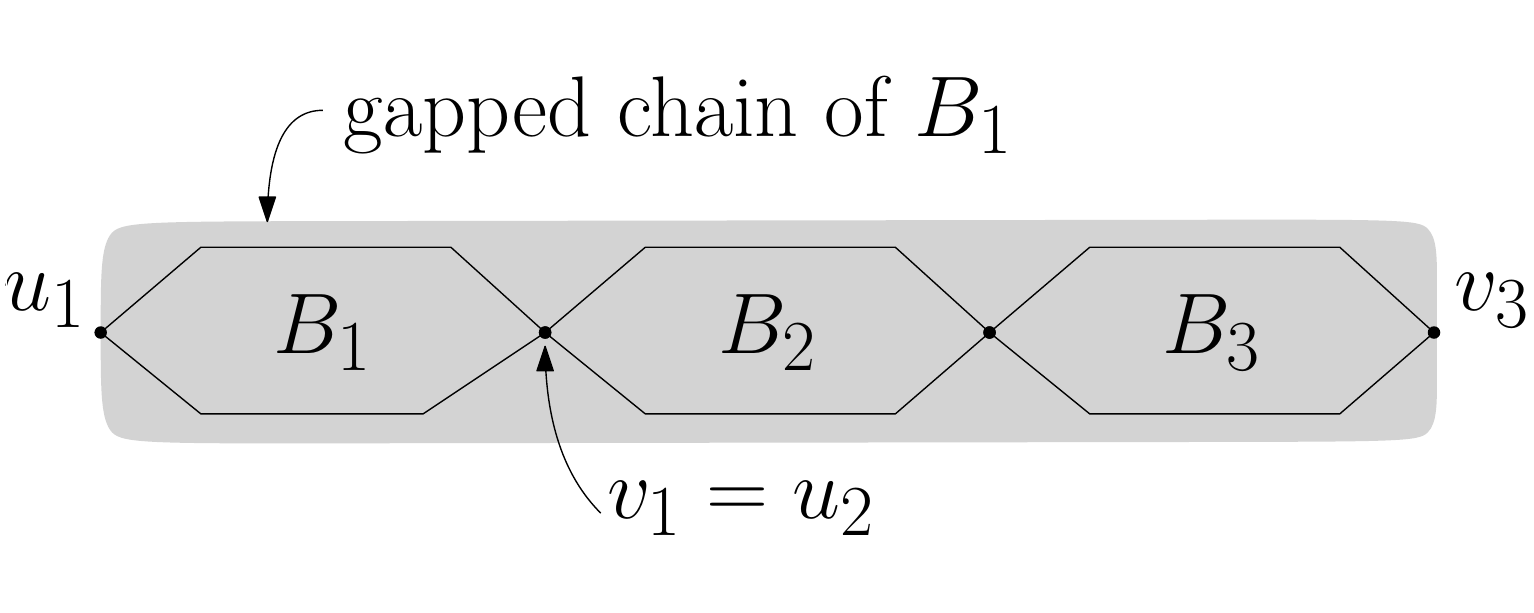}\label{fig:path-gadget}}\hfil
 \subfigure[$R$]{\includegraphics[width=0.45\columnwidth, page=2]{figs/gadget}\label{fig:ring}}
 \subfigure[$G_A$]{\includegraphics[width=\columnwidth, page=3]{figs/gadget}\label{fig:reduction}}
 \caption{(a) Schematization of a path gadget $P_3$. (b) A global ring $R$. (c) Schematization of the instance $G_A$ with  $m = 3$, $A = \{7, 7, 7, 8, 8, 8, 8, 9, 10\}$
  and $I = 24$. Transversal paths are routed according to the following solution of \textsc{3Partition} $A_1 = \{7, 7, 10\}$, $A_2 = \{7, 8, 9\}$ and $A_3 = \{8, 8, 8\}$. For simplicity, the gapped chains of the various blobs are not shown, as well as vertex $w$ and all the degree-$2$ vertices of the transversal paths. }
\end{figure}

We start the construction of $G_A$ with a global ring $R$. Let $\alpha$, $\beta$, $\gamma$, $\delta$ be the attaching vertices shared by $G_l$ and $G_t$, $G_t$ and $G_r$, $G_r$ and $G_b$, $G_b$ and $G_l$, respectively (see also Fig.~\ref{fig:ring}). First we add the edges $(\alpha,\beta)$ and $(\gamma,\delta)$. Denote as $R^+$ the resulting graph, and consider a \oneblip drawing of this graph. The gapped chain of $R$ subdivides the plane into a set of connected regions, such that only two of them contain all of $\alpha$, $\beta$, $\gamma$, and $\delta$ on their boundaries. We denote these two regions as $r_1$ and $r_2$. For ease of illustration, we assume that one of them is infinite (as in Fig.~\ref{fig:ring}), say $r_2$. Since the drawing is \oneblip, each of $(\alpha,\beta)$ and $(\gamma,\delta)$ is drawn entirely  in one of these two regions. We assume that both these two edges are drawn in the same region, say $r_2$, and we will later show that this is the only possibility in any \oneblip drawing of the final graph $G_A$.

We continue by connecting  the top and bottom beams by a set of $3m$ \emph{columns}; refer to Fig.~\ref{fig:reduction}. We describe each column in terms of its drawing, and we will later see that this is the only possible drawing that can be part of a \oneblip drawing of $G_A$. A column consists of $2m - 1$ \emph{cells}; a cell consists of a set of pairs of  crossing edges, called its \emph{vertical pairs}. Cells of the same column are separated by $2m-2$ path gadgets, called \emph{floors}. In particular, the cell between the $(m-1)$-st floor and the $m$-th floor is called the \emph{central cell}. Again, we are assuming a particular left-to-right order for the attaching vertex of a floor, and we will see that this is the only possible order in a \oneblip drawing.  The central cells (we have $3m$ of them in total) have a number of vertical pairs depending on the elements of $A$. Specifically, the central cell $C_i$ of the $i$-th column contains $a_i$ vertical pairs connecting its delimiting floors ($i \in \{1, 2, \dots, 3m\}$). Each of the remaining cells each has $\lceil I/2 \rceil +1$ vertical pairs. Hence, a noncentral cell contains more edges than any central cell. Further, the number of attaching vertices of a floor can be computed based on how many vertical pairs must be connected to the gadget.

It is now straightforward to see that it is not possible to draw both a column and one of $(\alpha,\beta)$ and $(\gamma,\delta)$ in $r_1$ or $r_2$ without violating {\oneblip}ity. Hence, we shall assume that both $(\alpha,\beta)$ and $(\gamma,\delta)$ are in $r_2$ and that all the columns are in $r_1$. Consider now a \oneblip drawing of a column. If we invert the left-to-right order of the attaching vertices of a floor (i.e., we mirror its drawing), then the resulting drawing is not \oneblip, since each floor delimits at least one noncentral cell with $\lceil I/2 \rceil +1$ vertical pairs. Moreover, since each vertical pair has a gapped edge, two vertical pairs cannot cross each other in a \oneblip drawing, and thus the drawings of the columns are disjoint from one another.

Finally, let $a$ and $b$ be the attaching vertices of the left and right walls distinct from $\alpha$, $\beta$, $\gamma$, and $\delta$. We connect $a$ and $b$ with $m$ pairwise internally disjoint paths, called \emph{transversal paths}; each transversal path has exactly $(3m - 3)(\lceil I/2 \rceil +1) + I$ edges. The routing of these paths will be used to determine a solution of $A$, if it exists. Thus, we aim at forcing the transversal paths to be inside $r_1$ in a \oneblip drawing of the graph. For this purpose, adding a vertex $w$ connected to all the attaching vertices of $G_t$ and $G_b$ will suffice.
Due to the presence of the columns in $r_1$, vertex $w$ must be in $r_2$ and, due to the edges $(\alpha,\beta)$ and $(\delta,\gamma)$ in $r_2$, all its incident edges (except at most two) are gapped. Thus, the transversal paths must be drawn in $r_1$. This concludes the construction of $G_A$.

\ShoLong{
}{
 We can prove the following.
}

\begin{theorem}\label{th:hardness}
 The \textsc{1GapPlanarity} problem is \textsc{NP}-complete.
\end{theorem}
\newcommand{\proofOneGapPlanarity}{
 \begin{proof}
  The \textsc{1GapPlanarity} problem is in \textsc{NP} by Lemma~\ref{le:np}.

  We now prove that an instance $A$ of \textsc{3Partition} is a positive instance if and only if the graph $G_A$ is a positive instance of \textsc{1GapPlanarity}.

  Suppose first that $G_A$ is a positive instance of \textsc{1GapPlanarity}. From the above discussion, it is clear that each traversing path must be routed through exactly three central cells and $3m - 3$ noncentral cells. In particular, each path has $(3m - 3)(\lceil I/2 \rceil +1) + I$ edges, and hence can traverse at most these many vertical pairs. Since each noncentral cell consists of $\lceil I/2 \rceil +1$ vertical pairs, it must be that the 3 central cells contain $I$ vertical pairs in total. Thus, we can construct a solution for $A$ by looking at the central cells traversed by the $m$ paths.

  Suppose now that $A$ is a positive instance of \textsc{3Partition}. Note that a \oneblip drawing of $G_A$ can be always computed if one omits all the transversal paths (see also Fig.~\ref{fig:reduction}).  To draw the paths, let $\{A_1, A_2,..., A_m\}$ be a solution for $A$. Then we route the paths similarly as in~\cite{Bekos2016}, that is, in such a way that: (1) they do not cross each other; (2) they do not cross any blob; (3) each path passes through exactly 3 central cells with $I$ vertical pairs in total, and $3m - 3$ noncentral cells; and (4) each cell is traversed by at most one path.
  Consider a subset $A_j$ of the solution of instance $A$ of \textsc{3Partition} and assume without loss of generality that $A_j=\{a_\kappa, a_\lambda,a_\mu\}$, where $1 \leq \kappa, \lambda, \mu \leq 3m$. Then, in the computed drawing,  path $\pi_j$ will cross the $\kappa$-th, $\lambda$-th and $\mu$-th columns of $G_A$ through central cells, while it will cross the remaining columns of $G$ through noncentral cells. Hence, requirement (3) is satisfied.
  Consider now the routing of the remaining transversal paths through the $\kappa$-th column; the corresponding routings though the $\lambda$-th and $\mu$-th columns of $G_A$ are symmetric. By construction, there must exist exactly $m-1$ available cells above and exactly $m-1$ available cells below the central cell of the $\kappa$-th column. This implies that there exist at least as many available noncentral cells as transversal paths to route at each side of the central cell of the $\kappa$-th column. Hence, we can route the remaining transversal paths through the $\kappa$-th column so that all other requirements are  satisfied.
 \end{proof}
}
\ShoLong{
 \proofOneGapPlanarity
}{
}

We conclude by observing that our proof can be easily adjusted for the setting in which the rotation system of the input graph is fixed. We call this problem \textsc{1GapPlanarityWithRotSys}. It suffices to choose a rotation system for $G_A$ that guarantees the existence of a \oneblip drawing ignoring the transversal paths (we already discussed the details of this drawing), and such that the transversal paths are attached to $a$ and $b$ with the ordering of their edges around $a$  reversed with respect to the ordering around $b$. The membership of the problem to \textsc{NP} can be easily verified (similarly to Lemma~\ref{le:np}). Thus, the next theorem follows.

\begin{theorem}\label{th:hardness-fr}
 The \textsc{1GapPlanarityWithRotSys} problem is \textsc{NP}-complete.
\end{theorem}

Note that our proof does not distinguish between simple and nonsimple drawings: Theorems~\ref{th:hardness} and~\ref{th:hardness-fr} work for both cases. In fact, if the rotation system is not fixed, a graph is 1-gap-planar if and only if it admits a simple 1-gap-planar drawing (self-crossings and multiple crossings can be redrawn as explained in Lemma~\ref{lem:simpletop}). When the rotation system is fixed the statement is not always true. This is due to the fact that the redrawing may alter the rotation system. Thus, it is possible that a graph has a nonsimple 1-gap-planar drawing for some rotation system, while it does not admit a simple 1-gap-planar drawing with the same rotation system (in such a case, it must admit a simple drawing with a different rotation system).

\section{Conclusions and open problems}\label{sec:conclusions}

\begin{figure}
 \centering
    \subfigure[]{\includegraphics[width=0.35\columnwidth]{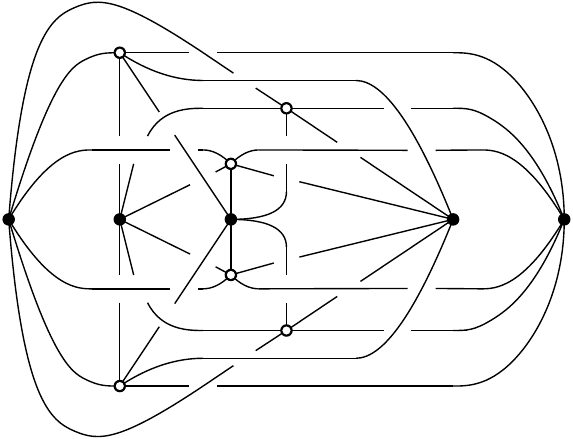}\label{fig:k56}}\hfil
\subfigure[]{\includegraphics[width=0.39\columnwidth]{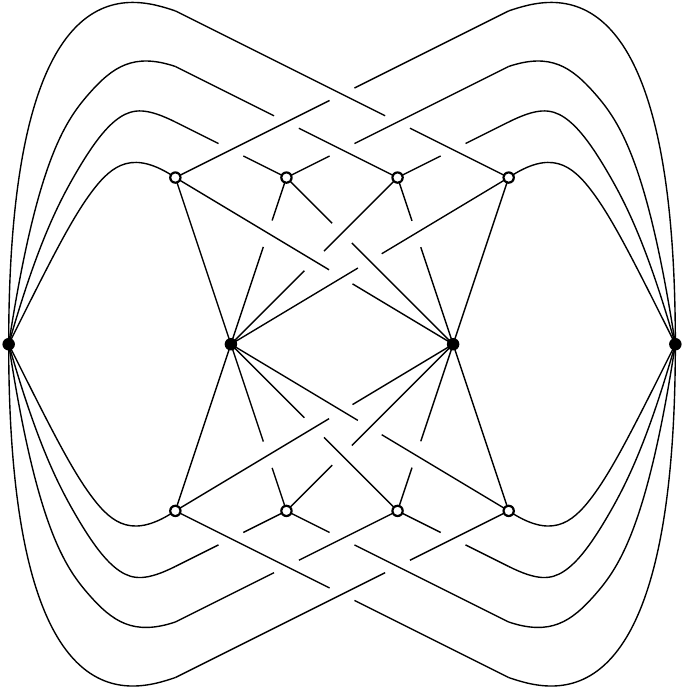}\label{fig:k48}}
 \caption{\oneblip drawings of $K_{5,6}$ (left) and $K_{4,8}$ (right).}
\end{figure}

We introduced \kblip graphs, and our results give rise to several questions for future research. Among them are:

\begin{enumerate}[(i)]
\item In Theorem~\ref{thm_crossing} we characterized \kblip graphs by an adaptation of Hall's condition. We wonder whether a similar characterization is possible based on the crossing number. Specifically, does there exist some function $f:\mathbb{N}\rightarrow \mathbb{N}$ such that if ${\rm cr}(G')\leq f(k) |E(G')|$ for every subgraph $G'$ of a graph $G$, then $G$ is \kblip?

 \item We proved that \kblip graphs with $n$ vertices have $O(\sqrt{k} \cdot n)$ edges, which is tight apart from constant factors; and that \oneblip graphs have at most $5n-10$ edges, which is a tight bound for $n\geq 20$.
     Can one establish a tight bound also for \twoblip graphs?

\item We proved that a drawing with at most $2k$ crossings per edge is \kblip, and that a \kblip drawing does not contain $2k+2$ pairwise crossing edges.  Do \oneblip graphs have RAC drawings with at most $1$ or $2$ bends per edge? What is the relationship between \oneblip graphs and fan-planar graphs?

 \item We proved that $K_n$ is \oneblip if and only if $n \le 8$. A similar characterization could be studied also for complete bipartite graphs. Note that $K_{5,7}$ is not \oneblip since its crossing number is greater than its number of edges, while $K_{5,6}$ admits a \oneblip drawing (Fig.~\ref{fig:k56}).
     We do not know whether $K_{6,6}$ is \oneblip. Similarly, $K_{3,12}$ (Fig.~\ref{fig:k312}) and $K_{4,8}$ (Fig.~\ref{fig:k56}) 
    are \oneblip, but we do not know whether this is true also for $K_{3,13}$ and $K_{4,9}$.

 \item We proved that deciding whether a graph is \oneblip is \textsc{NP}-complete, even if the rotation system is fixed. Can the problem be solved in polynomial time for drawings in which all vertices are on the outer boundary?

\end{enumerate}

\section*{Acknowledgments}
This research started at the NII Shonan Meeting ``Algorithmics for Beyond Planar Graphs.'' The authors thank the organizers and all participants for creating a pleasant and stimulating atmosphere. In particular, we thank Yota Otachi for useful discussions. We also thank David Eppstein and David Wood for pointing out several useful references as well as the anonymous referees of a preliminary version of this paper for their insightful suggestions.

\bibliography{blipgraphs}
\bibliographystyle{plain}

\ShoLong{}{
\section*{Appendix}
 \section{\newpage

\appendix
Additional Material for Section~\ref{sec:density}}
 \label{ap:density}

 \proofLemmaOneSection

 \subsection{Lower bound constructions}

 \begin{figure}[h]
  \centering
  \includegraphics[page=3]{figs/5n}
  \caption{One more pattern that produces \oneblip graphs with $n$ vertices and $5n-\Theta(1)$ edges.}
  \label{fig:5n-d}
 \end{figure}

 \mentionLowerBounds

 \section{Additional Material for Section~\ref{sec:recognition}}\label{ap:recognition}

 \setcounter{lemma}{2}
 \begin{lemma}
 Let $u$ and $v$ be any two degree-$12$ vertices~of~$B$. Every \oneblip drawing $\Gamma$ of $B$ contains a gapped chain $\mathcal C$ containing $u$ and $v$.
 \end{lemma}
 \proofblippedcycle

 \lemmaNP

 \setcounter{theorem}{4}
 \begin{theorem}
  The \textsc{1GapPlanarity} problem is \textsc{NP}-complete.
 \end{theorem}
 \proofOneGapPlanarity

 \section{Additional Material For Section~\ref{sec:relationship}}\label{ap:relationship}

 \setcounter{lemma}{6}
 \begin{lemma}
  For every $k \ge 1$, there exists a \oneblip graph $G_k$ that is not $k$-planar.
 \end{lemma}
 \proofoneblipnotkplanar

\section{1-Gap-Planar Drawings of Complete Bipartite Graphs}
\label{sec:bipartite}

} 

\end{document}